\documentclass[11pt]{article}

\usepackage{fullpage}
\usepackage{setspace}
\usepackage{natbib}
\usepackage{amsmath}
\usepackage{amssymb}
\usepackage{mathtools}
\usepackage[utf8]{inputenc} 
\usepackage[T1]{fontenc}    
\usepackage{hyperref}       
\usepackage{url}           
\usepackage{booktabs}       
\usepackage{algorithm}
\usepackage{amsfonts}     
\usepackage{nicefrac}      
\usepackage{microtype}      
\usepackage[dvipsnames]{xcolor}        
\usepackage{graphicx}
\usepackage{float}
\usepackage{tikz}
\usetikzlibrary{matrix,positioning,fit,calc,backgrounds}
\usepackage{bm}
\usepackage{bbm}
\usepackage{amsthm}
\usepackage{mathrsfs}
\usepackage{multirow}
\usepackage{booktabs}
\usepackage{longtable}
\usepackage{tabularx}
\usepackage{graphicx}
\usepackage{subcaption}
\usepackage{caption}
\usepackage{algorithm}
\usepackage{algpseudocode}
\usepackage{arydshln}
\usepackage{makecell}
\usepackage{threeparttable}
\usepackage{adjustbox}

\definecolor{revisioncolor}{RGB}{0,0,0}
\newcommand{\revision}[1]{{\color{revisioncolor}#1}}
\newenvironment{revisionblock}{\begingroup\color{revisioncolor}}{\endgroup}

\newcommand{\Var}{\operatorname{Var}}
\newcommand{\E}{\mathbb E}

\newtheorem{lemma}{Lemma}
\newtheorem{theorem}[lemma]{Theorem}
\newtheorem{proposition}[lemma]{Proposition}
\newtheorem{definition}{Definition}
\newtheorem{assumption}{Assumption}
\newtheorem{remark}{Remark}

\theoremstyle{definition}
\newtheorem{example}{Example}

\title{Estimating Dynamic Infinitesimal Policy Effects\\
under Sequential Unconfoundedness}

\author{
I-han Lai\\
Mgmt.\ Science \& Engineering\\
Stanford University\\
\texttt{ihan@stanford.edu}
\and
Stefan Wager\\
Graduate School of Business\\
Stanford University\\
\texttt{swager@stanford.edu}
}
\date{}

\begin{document}

\maketitle

\begin{abstract}
We develop methods for estimating how small policy changes affect long-term outcomes in dynamic systems. We show that a class of dynamic \revision{infinitesimal policy effects (IPEs)} can be identified via tractable reduced-form expressions, and can be estimated under a general sequential unconfoundedness assumption. We also propose a doubly robust estimator for these policy effects. Our approach does not require observing full dynamic state information (as is typically assumed for off-policy evaluation in Markov decision processes), and does not incur an exponential curse of horizon (as is typical in non-Markovian off-policy evaluation). We demonstrate practicality and robustness of our approach in a number of simulations, including one motivated by a dynamic pricing application where people use past prices to form a reference level for current prices.
\end{abstract}


\section{Introduction}

Dynamic decision-making problems,\footnote{This research was partially funded by grant N00014-24-1-2091 from the Office of Naval Research.}
with actions taken over time in response to evolving information, arise in many domains including online platforms, public policy, and personalized medicine
\citep{sutton1998reinforcement,tsiatis2019dynamic}. In such settings, there is often interest in performing off-policy evaluation, i.e., in using data collected under a status-quo policy to reason about how key outcomes (e.g., long-term revenue, welfare or health metrics) might change if we used a different counterfactual policy to guide action. Dynamic off-policy evaluation is, however, generally challenging.
Methods that exploit Markov structure require a suitable observed state,
whereas efficiency bounds in general history-dependent models can grow
exponentially with the horizon
\citep{dudik2014doubly,liao2022batch,jiang2016doubly,kallus2022efficiently,thomas2016data,kallus2020double}.

In this paper, we argue that restricting our attention to counterfactual evaluation of policies that are ``local perturbations'' to the status-quo policy enables analysts to meaningfully learn about
policy-relevant treatment effects while avoiding many of the difficulties involved in standard off-policy evaluation. Specifically, we consider estimation of a class of dynamic infinitesimal
policy effects (IPEs), which we define as the pathwise derivatives of expected discounted welfare along
smooth user-specified perturbations of a baseline decision rule. Conceptually, the IPE compares the value of the current policy to the value of a nearby policy obtained by infinitesimally
perturbing the action-assignment mechanism in a pre-specified direction. For example, the IPE can be used to quantify the effect of slightly increasing probability of a given action for a particular
subgroup in a dynamic treatment regime, or of nudging a pricing rule upward under specific market
conditions. Our main finding is that the dynamic IPE can be identified via simple non-recursive moment conditions, and allows for doubly robust
estimation under sequential unconfoundedness using algorithms whose complexity is comparable to the augmented inverse-propensity weighted algorithm of \citet{robins1994estimation} and its
generalizations \citep{chernozhukov2022locally}. We also demonstrate practical
feasibility of our approach in a number of simulation experiments, including a design
motivated by optimal pricing in a dynamic setting where consumers use past prices as
reference points when reacting to current prices.

The promise of focusing on estimating the effect of infinitesimal interventions has been widely recognized
across a number of application areas. \citet{heckman2005structural} and \citet{carneiro2010evaluating}
argue that, in models with endogenous selection into treatment, marginal policy effects (MPEs)---which can be interpreted as IPEs scaled by the number of people affected by the policy change---present a realistic
yet still policy-relevant target for inference. \citet{kennedy2019nonparametric} advocates
for evaluating interventions that modify propensity scores (as opposed to ones that fix treatment paths) as a pragmatic approach to dynamic treatment rules. \citet{munro2025treatment} consider
causal inference in marketplaces under supply-demand equilibrium, and find that various MPEs can be estimated---and used for welfare-improving targeting---using simple non-parametric approaches.
And \citet{wager2021experimenting} argue that, beyond being of inherent interest, IPEs can be used to efficiently improve overall system performance via gradient-based algorithms.

Relative to this line of work, our main contribution is to develop general results for estimating infinitesimal effects that apply to dynamic problems. Most existing work on such estimands
focuses on single time period problems with difficulties arising from endogenous selection and/or equilibrium formation. Estimating dynamic IPEs poses new challenges because the
effects of any policy change propagate through time and so, e.g., a decision taken today not only affects the outcome observed today but also shapes the distribution of
covariates and actions in subsequent time periods. This phenomenon is closely related to carryover effects \citep{bojinov2023design,hu2022switchback} and temporal interference \citep{farias2022markovian,glynn2020adaptive}.

The starting point for our analysis is a reduced-form representation
of dynamic IPEs using a generalization of the widely used policy-gradient theorem from
reinforcement learning \citep[e.g.,][]{marbach2001simulation,NIPS1999_464d828b}.
This reduced-form representation then
allows us to approach estimation and inference via the general doubly robust / debiased machine learning framework \citep{chernozhukov2022locally}.
The policy-gradient theorem is widely used for optimizing reinforcement learning systems
via gradient-based algorithms. However, the use of policy-gradient methods in
observational-study causal inference is as-of-now still an emerging area; two contributions we are aware of are the work of \citet{johari2025estimation} on policy-gradient estimation in Bernoulli-randomized experiments, and \citet{ghosh2025non} who study
dynamic regression-discontinuity designs via the lens of MPE estimation.

\subsection{Related Work}
\label{sec:related_work}

The statistical literature on causal inference in longitudinal settings goes back to Robins' work on the $g$-formula and related $g$-methods, including marginal structural models (MSMs) and structural nested models, which address time-varying confounding and provide identification and estimation strategies for counterfactual trajectories under sequential interventions \citep{robins1986new, robins1997causal}. Under sequential unconfoundedness (or sequential ignorability), these frameworks imply a factorization of counterfactual distributions that enables estimation via iterated regression ($g$-computation), inverse-propensity weighting, or their augmentations \citep{hernan2020causal, tsiatis2019dynamic}.

A closely related literature studies dynamic treatment regimes (DTRs), i.e., policies mapping evolving histories to actions, motivated by multi-stage clinical decision-making \citep{murphy2003optimal,robins2004optimal,chakraborty2014dynamic}. Subsequent work leverages modern machine learning and high-dimensional statistics for regime evaluation and learning, including settings with complex longitudinal structure \citep{ertefaie2018constructing, bodory2022evaluating, bradic2024high}.
Meanwhile, in the reinforcement learning literature, there has been sustained interest in methods for off-policy evaluation that can be used to estimate the value of a target policy using data generated by a different behavior policy. Importance-sampling estimators are unbiased but can suffer from high variance, especially over long horizons; doubly robust estimators combine outcome modeling with weighting to reduce variance while retaining robustness \citep{dudik2014doubly, jiang2016doubly, thomas2016data}. Recent work also develops semiparametric and efficiency-oriented analyses of off-policy evaluation in Markov decision processes \citep{kallus2022efficiently, liao2022batch, mehrabi2024off}. 

\revision{Relative to this line of work, our contribution is to develop methods for estimating infinitesimal
effects, and to show that these methods have favorable statistical properties. In general, non-parametric
methods for off-policy evaluation have variance that scales exponentially in the decision
horizon $T$ \citep{kallus2020double}; and, while methods rooted in Markov decision processes can address this challenge, they
rely on the strong (and in some applications implausible) assumption that we can observe all states
that mediate carryover treatment effects. In contrast, the methods for infinitesimal effects developed
here suffer from neither of these difficulties: We can obtain stable estimates over moderately long
horizons without assuming that the system is Markovian.}

Our estimation strategy is informed by the literature on semiparametric estimation of complex functionals with high-dimensional nuisance components. The organizing principle is to construct Neyman-orthogonal (locally robust) scores so that first-order nuisance errors cancel, enabling valid inference under flexible nuisance learning \citep{chernozhukov2018double}. Automatic Debiased Machine Learning (Auto-DML) characterizes the correction term via a Riesz representer and proposes practical estimation strategies based on variational formulations \citep{chernozhukov2022automatic}. Related balancing and minimax linear weighting approaches target the same robustness goal through finite-sample control of imbalance \citep{hirshberg2021augmented}. The present paper extends this debiasing logic to local perturbations in dynamic treatment regimes.

Several literatures study welfare changes induced by small policy shifts relative to the status quo. In econometrics, the marginal treatment effect framework and related policy parameters characterize the welfare consequences of marginal changes in participation incentives or assignment rules \citep{heckman2005structural, carneiro2010evaluating, sasaki2023prte}. These contributions are conceptually aligned with our focus on local policy perturbations, but they primarily address static settings where key challenges are due to selection on unobservables. Our estimand instead concerns a sequential decision problem and uses a policy-gradient representation to express the infinitesimal effect as a linear functional of stage-wise $q$-functions.

Meanwhile, in biostatistics, a growing body of work targets effects of feasible interventions that modify the treatment mechanism rather than deterministically setting treatment values. This includes modified/shift interventions for continuous or multivalued treatments \citep{haneuse2013estimation}, incremental propensity score interventions \citep{kennedy2019nonparametric, naimi2021incremental}, and extensions tailored to resource constraints and time-varying treatment availability \citep{diaz2023lmtp, sarvet2023longitudinal}. These approaches share the motivation of defining policy-relevant contrasts that mitigate positivity challenges and better reflect implementable changes. Our work complements this line by focusing on an infinitesimal perturbation of a baseline policy over its entire horizon; and establishing that the resulting marginal policy effect admits a policy-gradient representation that directly motivates a doubly robust estimation strategy.

Finally, there is a long tradition of using policy-gradient methods to sequentially
optimize parameters within a reinforcement learning system
\citep{williams1992simple, NIPS1999_464d828b, NIPS2001_4b86abe4}.
However, this paper is one of just a handful of recent investigations around the use of
policy-gradient methods for observational study causal inference. Others include
\citet{johari2025estimation} who consider estimation of policy gradients from sequentially
Bernoulli-randomized trials; and \citet{ghosh2025non} who study local welfare sensitivity for
structured policy classes such as threshold rules.

\section{Dynamic Infinitesimal Policy Effects}
\label{sec:setup}

We study a finite-horizon sequential decision problem with horizon $T\in\mathbb{N}$, observed over $n$ independent and identically distributed units. For unit $i$, the observed trajectory is
\[
Z_i \;=\; (X_{i,1},A_{i,1},Y_{i,1},\,X_{i,2},A_{i,2},Y_{i,2},\,\ldots,\,X_{i,T},A_{i,T},Y_{i,T}),
\]
where $X_{i,t}\in\mathcal X_t$ denotes covariates observed at the beginning of period $t$, $A_{i,t}\in\mathcal A_t$ is the action taken, and $Y_{i,t}\in\mathbb{R}$ is the accrued outcome. For notational simplicity, we suppress the unit index and work with a generic trajectory; we will also use a convention $X_{i,T+1} = \emptyset$ to simplify some expressions.
It is convenient to collect the information observed before action $A_t$ is chosen into the history
\[
S_t \;:=\; (X_1,A_1,Y_1,\ldots,X_{t-1},A_{t-1},Y_{t-1},X_t)\in\mathcal S_t,
\]
so that decisions may depend on the entire past. We do not impose a Markov restriction: the conditional law of $(X_{t+1},Y_t)$ given $(S_t,A_t)$ may depend on the full history in an arbitrary manner.

A (possibly stochastic) policy is a collection $\pi=\{\pi_t\}_{t=1}^{T}$ of conditional action distributions. Throughout, we focus on history-dependent policies indexed by the current history,
\begin{equation}
A_t \sim \pi_t(\cdot\mid S_t), \qquad t=1,\ldots,T.
\label{eq:policy_def}
\end{equation}
We fix a single $\sigma$-finite reference measure $\lambda$ on the action space $\mathcal A$ such that each $\mathcal A_t$ is a measurable subset of $\mathcal A$ and, for almost every history $s$, the conditional distribution $\pi_t(\cdot\mid s)$ admits a density (or mass function) with respect to $\lambda$ restricted to $\mathcal A_t$. 

For a discount factor $\gamma\in(0,1]$, define the discounted cumulative outcome
from period $t$ onward as
\begin{equation}
\Gamma_t
:=
\sum_{k=t}^T \gamma^{k-t}Y_k.
\label{eq:defY}
\end{equation}
Under any policy $\pi$, define the action-value and value functions
\begin{align}
q_t(s,a)
&:=
\E_\pi\!\left[
\Gamma_t
\,\middle|\,
S_t=s,\, A_t=a
\right],
\label{eq:def_q}
\\
V_t(s)
&:=
\int_{\mathcal A_t} q_t(s,a)\pi_t(a\mid s)\,d\lambda(a),
\qquad
V_{T+1}\equiv0.
\label{eq:def_V}
\end{align}
We now formalize the longitudinal causal model and the support conditions under which policy interventions are represented via the $g$-formula. The following two assumptions characterize the status-quo data collection distribution.

\begin{assumption}[Consistency]
\label{assump:consistency}
For each action history $a_{1:t}=(a_1,\ldots,a_t)$ there exist potential outcomes
\(
(X_{t+1}(a_{1:t}),Y_t(a_{1:t}))
\)
such that, almost surely,
\[
(X_{t+1},Y_t) = (X_{t+1}(A_{1:t}),Y_t(A_{1:t})), \qquad t=1,\ldots,T.
\]
\end{assumption}

\begin{assumption}[Sequential unconfoundedness]
\label{assump:SU}
For each $t=1,\ldots,T$ and any future action sequence $(a_t,\ldots,a_T)$,
\[
A_t \ \perp\!\!\!\perp\ \Big\{
X_{k+1}(A_{1:t-1},a_t,\ldots,a_k),\,
Y_k(A_{1:t-1},a_t,\ldots,a_k)
:\ k=t,\ldots,T
\Big\} \ \Big|\ S_t.
\]
Equivalently, conditional on $S_t$, the realized action $A_t$ is independent of all future potential outcomes that are consistent with the realized past $A_{1:t-1}$.
\end{assumption}

Writing $P(x_1)$ for the marginal law of $X_1$ under the observed data-collection distribution
and $P(x_{t+1}, y_t\mid s_t,a_t)$ for the corresponding conditional law of $(X_{t+1},Y_t)$ given
$(S_t,A_t)=(s_t,a_t)$ for each $t=1,\ldots,T$, the status quo data-generating distribution factors as
\begin{equation}
P(x_{1:T},a_{1:T},y_{1:T})
=
P(x_1)\prod_{t=1}^T P(a_t\mid s_t)\,P(x_{t+1},y_t\mid s_t,a_t),
\end{equation}
where $s_t=(x_1,a_1,y_1,\ldots,x_{t-1},a_{t-1},y_{t-1},x_t)$,
$P(a \mid s)$ is interpreted as the status-quo policy, and
as usual $x_{T+1} = \emptyset$.
Our main interest is in how the outcome distribution would change if we were to
use any counterfactual policy $\pi=\{\pi_t\}_{t=1}^T$ for choosing actions. Under Assumptions~\ref{assump:consistency} and \ref{assump:SU},
trajectories generated using any policy $\pi$ whose conditional
action distribution is absolutely continuous with respect to the status-quo policy
satisfy the $g$-formula of \citet{robins1986new}, i.e.,
the induced counterfactual joint law on $(X_{1:T},A_{1:T},Y_{1:T})$ is
\citep[see, e.g.,][Proposition 14.1]{wager2024causal}
\begin{equation}
\label{eq:gformula}
P^{\pi}(x_{1:T},a_{1:T},y_{1:T})
=
P(x_1)\prod_{t=1}^T \pi_t(a_t\mid s_t)\, P(x_{t+1},y_t\mid s_t,a_t).
\end{equation}
In particular, only the action assignment rule varies across policies; the law of $X_1$ and the transition kernel $P(\cdot\mid S_t,A_t)$ are shared across all policies. \revision{(The $g$-formula directly identifies the joint distribution of observations generated under policies $\pi_t(\cdot \mid s_t)$ that share support with the baseline policy. Our framework will also be applicable to some perturbations where, e.g., due to added noise, some actions may leak beyond the support of the status quo; see Remark \ref{rem:first_order_support} for further discussion.)}

In many practical settings---such as a platform applying a price markup, a clinician adjusting a dosage, or an algorithm injecting exploration noise---interventions are most naturally defined as direct modifications to status-quo decisions.
We formalize this insight through kernel-generated convolutional perturbation paths, including paths induced by continuous-time Markov processes.
At stage $t$ and given history $s$, let $a_0\sim \pi_t(\cdot\mid s)$ denote the action drawn under the status-quo or baseline policy. \revision{A convolutional perturbation modifies the baseline action through a Markov kernel $M_{t,\varepsilon}(\cdot\mid a_0,s)$, where $\varepsilon\ge0$ measures the size of the local modification.}

\begin{definition}[Convolutional perturbation path]
\label{def:operator_generated}
Fix a baseline policy $\pi=\{\pi_t\}_{t=1}^T$. A family of policies
$\{\pi_\varepsilon\}_{0\le \varepsilon<\bar\varepsilon}$ is a convolutional
perturbation path around $\pi$ if, for each stage $t=1,\ldots,T$, there exists a family of Markov kernels
$M_{t,\varepsilon}(\cdot\mid a_0,s)$ for $0\le \varepsilon<\bar\varepsilon$,
mapping baseline actions $a_0\in\mathcal A_t$ to perturbed actions
$a\in\mathcal A_t$, such that the following conditions hold.

\begin{enumerate}
\item[(i)] \textbf{Kernel-generated policy.}
The perturbed policy is the marginal distribution obtained by drawing
$a_0\sim\pi_t(\cdot\mid s)$ and then applying the kernel $M_{t,\varepsilon}(\cdot\mid a_0,s)$:
\begin{equation}
\pi_{t,\varepsilon}(\cdot\mid s)
=
\int_{\mathcal A_t}
M_{t,\varepsilon}(\cdot\mid a_0,s)\,\pi_t(a_0\mid s)\,d\lambda(a_0).
\label{eq:operator_family}
\end{equation}
The generated policies are assumed to have densities or mass
functions with respect to $\lambda$, denoted by
$\pi_{t,\varepsilon}(a\mid s)$.

\item[(ii)] \textbf{Baseline recovery.}
At $\varepsilon=0$, the kernel leaves the baseline action unchanged:
\[
M_{t,0}(\cdot\mid a_0,s)=\delta_{a_0}.
\]
Consequently, $\pi_{t,0}(a\mid s)=\pi_t(a\mid s)$.

\item[(iii)] \textbf{Forward operator.}
The path admits a forward operator $L_t^*$ acting on the baseline policy density:
\begin{equation}
\label{eq:forward}
(L_t^*\pi_t)(a\mid s)
:=
\lim_{\varepsilon\downarrow0}
\frac{\pi_{t,\varepsilon}(a\mid s)-\pi_t(a\mid s)}
{\varepsilon}.
\end{equation}

\item[(iv)] \textbf{First-order support.}
For almost every state $S_t$, the forward operator defined in (iii) is absolutely continuous with respect to the baseline policy:
\[
\pi_t(a\mid S_t)=0
\quad\Longrightarrow\quad
(L_t^*\pi_t)(a\mid S_t)=0
\quad\text{for $\lambda$-almost every }a\in\mathcal A_t.
\]
\end{enumerate}
\end{definition}

\begin{revisionblock}
Under the regularity conditions below, the forward operator \eqref{eq:forward} also
admits a dual backward operator that acts on suitably smooth action-value functions.
Throughout, the class $\mathcal Q_t$ is taken to be as constructed below.

\begin{lemma}[Backward operator]
\label{lemm:Qclass}
Fix $t\in\{1,\ldots,T\}$ and suppose
Definition~\ref{def:operator_generated} holds.
Let $\mathcal Q_t$ be a linear class of bounded measurable
functions on $\mathcal S_t\times\mathcal A_t$ admitting the following right
derivative in $L^1$ under the baseline law of $(S_t,A_t)$:
\begin{equation}
L_tq
:=
\left.
\frac{d}{d\varepsilon}\E_{A\sim M_{t,\varepsilon}(\cdot\mid a_0,s)}
[q(s,A)]
\right|_{\varepsilon=0+} \ \text{ for all } \ \ q\in\mathcal Q_t.
\label{eq:backward}
\end{equation}
Suppose further that
\[
\lim_{\varepsilon\downarrow0} \E\!\left[
\int_{\mathcal A_t}
\left|
\frac{\pi_{t,\varepsilon}(a\mid S_t)-\pi_t(a\mid S_t)}
{\varepsilon}
-
(L_t^*\pi_t)(a\mid S_t)
\right|
\,d\lambda(a)
\right]
 = 0.
\]
Then $L_t$ is linear on $\mathcal Q_t$, and every
$q\in\mathcal Q_t$ satisfies
\[
\E\!\left[(L_tq)(S_t,A_t)\mid S_t\right]
=
\int_{\mathcal A_t}
q(S_t,a)(L_t^*\pi_t)(a\mid S_t)\,d\lambda(a)
\qquad\text{almost surely}.
\]
\end{lemma}

\begin{remark}[Overlap]
\label{rem:first_order_support}
The first-order support condition {\it (iv)} is a relaxation of the standard overlap
condition for off-policy evaluation that requires the target action distribution to
be absolutely continuous with respect to the behavior policy
\citep{thomas2016data,kallus2020double,mehrabi2024off}.
The reason we consider this relaxation here is that, in many settings of interest,
convolutional perturbations may leak actions outside the baseline support, e.g., if
we add Gaussian noise to a baseline distribution with bounded support---and it turns
out that our approach can still be used to identify IPEs in this setting (even though
the effects of non-infinitesimal perturbations may not be identified). The key insight,
established by \citet{carneiro2010evaluating} in the static setting, is that if the
probability of support violation is $o(\varepsilon)$ for an $\varepsilon$-scale perturbation
then support violation may not preclude identification of IPEs. The first-order support
condition above is a quantitative version of this ``negligible support violation'' assumption.
\end{remark}
\end{revisionblock}

Our parameter of interest is the local change in expected discounted welfare at the baseline policy when we move along a convolutional perturbation path.
For a given path, let $\E_\varepsilon[\cdot]$ denote the expectation under the joint law induced by $\pi_\varepsilon$. The infinitesimal policy effect (IPE) is defined as
\begin{equation}
\Theta
:=
\left.
\frac{d}{d\varepsilon}
\E_\varepsilon\!\left[
\sum_{k=1}^T \gamma^{k-1}Y_k
\right]
\right|_{\varepsilon=0+}.
\label{eq:estimand_def_main}
\end{equation}
We next give an identification result for IPEs (i.e., a characterization of the IPE in terms of functionals of the observed baseline data-collection distribution) along convolutional perturbation paths. To state the result, define the weight function
\[
H_t(s,a)
:=
\frac{(L_t^*\pi_t)(a\mid s)}{\pi_t(a\mid s)}
\qquad
\text{on }\{\pi_t(a\mid s)>0\},
\]
where $L_t^*$ is as in \eqref{eq:forward}.
Set $H_t(s,a)=0$ elsewhere.
We also write
\begin{equation}
\begin{split}
&\langle L_t^*\pi_t,q\rangle(s)
:=
\int_{\mathcal A_t}
q(s,a)(L_t^*\pi_t)(a\mid s)\,d\lambda(a), \\
&\langle L_t q,\pi_t\rangle(s)
:=
\int_{\mathcal A_t}
(L_t q)(s,a)\pi_t(a\mid s)\,d\lambda(a).
\end{split}
\end{equation}
for functions $q$ in the domain of $L_t$ as given in \eqref{eq:backward}.

\begin{revisionblock}

\begin{theorem}[Convolutional policy-gradient theorem]
\label{thm:gPGT}
Suppose Assumptions~\ref{assump:consistency} and
\ref{assump:SU} hold, and let
$\{\pi_\varepsilon\}_{0\leq\varepsilon<\bar\varepsilon}$
be a convolutional perturbation path satisfying the conditions
of Lemma~\ref{lemm:Qclass}.
Assume that, for some finite constant $C$, under the baseline
and all sufficiently small perturbed interventions,
\[
|\Gamma_t|\leq C
\quad\text{almost surely},
\qquad t=1,\ldots,T,
\]
and that each baseline action-value function admits an extension
$q_t\in\mathcal Q_t$ with $\mathcal Q_t$ as in Lemma~\ref{lemm:Qclass}.
Then the infinitesimal policy effect in \eqref{eq:estimand_def_main}
exists and has the forward representation
\begin{equation}
\Theta
=
\sum_{t=1}^T
\gamma^{t-1}
\E\!\left[
\langle L_t^*\pi_t,q_t\rangle(S_t)
\right]
=
\sum_{t=1}^T
\gamma^{t-1}
\E\!\left[
H_t(S_t,A_t)\Gamma_t
\right],
\label{eq:gPGT}
\end{equation}
and the backward representation
\begin{equation}
\Theta
=
\sum_{t=1}^T
\gamma^{t-1}
\E\!\left[
\langle L_tq_t,\pi_t\rangle(S_t)
\right]
=
\sum_{t=1}^T
\gamma^{t-1}
\E\!\left[
(L_tq_t)(S_t,A_t)
\right].
\label{eq:gPGT2}
\end{equation}
\end{theorem}
\end{revisionblock}

Given access to samples under the status quo policy, evaluating $\Theta$ via \eqref{eq:gPGT}
further requires knowledge of the behavior policy $\pi_t(a\mid s)$, or equivalently the weights $H_t(s,a)$;
whereas the second representation \eqref{eq:gPGT2} requires $q_t(s, a)$. For purposes of estimation,
these dual representations of $\Theta$ will enable construction of a doubly robust moment condition
following \citet{chernozhukov2022locally}, thus enabling robust estimation of $\Theta$ via the debiased machine learning paradigm.

Before proceeding with estimation results, we end this section by instantiating the abstract
identification results given above in the context of some examples of notable types
of action-perturbations one may want to evaluate.

\begin{example}[\textbf{Action swaps}]
\label{ex:benchmark_policy}
The standard policy-gradient theorem quantifies the welfare effects of
infinitesimal changes in action probabilities over a discrete action set
$\mathcal A_t$.
Our convolutional policy-gradient theorem recovers this result using
Markov jump processes.
Consider a process that, at times determined by a Poisson process with
rate one, draws an action $a'$ from $K_t(\cdot\mid a,s)$.
Its backward operator is \citep[Chapter 20]{levin2017markov}
\[
(L_t q)(s,a)
=
\sum_{a'\in\mathcal A_t}
\left(q(s,a')-q(s,a)\right)K_t(a'\mid a,s).
\]
Define the action distribution after one update by
\[
k_t(a\mid s)
:=
\sum_{a_0\in\mathcal A_t}
K_t(a\mid a_0,s)\pi_t(a_0\mid s).
\]
The corresponding forward derivative is
\[
(L_t^*\pi_t)(a\mid s)
=
k_t(a\mid s)-\pi_t(a\mid s);
\]
note that here first-order support is thus equivalent to
the standard overlap condition for off-policy evaluation.
Under the remaining conditions of Theorem~\ref{thm:gPGT}, we thus get that the IPE is
\begin{align*}
\Theta
&=
\sum_{t=1}^T
\gamma^{t-1}
\E\!\left[
\sum_{a\in\mathcal A_t}
\left(q_t(S_t,a)-q_t(S_t,A_t)\right)
K_t(a\mid A_t,S_t)
\right] \\
&=
\sum_{t=1}^T
\gamma^{t-1}
\E\!\left[
\sum_{a\in\mathcal A_t}
q_t(S_t,a)
\left(k_t(a\mid S_t)-\pi_t(a\mid S_t)\right)
\right],
\end{align*}
where the second equality follows by conditioning on $S_t$ and
rearranging the sums.
This recovers the standard policy-gradient representation, as in
Theorem~1 of \citet{NIPS1999_464d828b}.
\end{example}

\begin{example}[\textbf{Deterministic shift}]
\label{ex:location_shift}
Consider continuous actions with 
$\mathcal A_t=\mathbb R$.
A uniform upward shift is generated by
\[
M_{t,\varepsilon}(\cdot\mid a_0,s)
=
\delta_{a_0+\varepsilon}.
\]
\begin{revisionblock}
Suppose that, conditionally on $S_t$, the baseline density almost surely admits a differentiable extension to the real line (i.e., $\pi_t(\cdot\mid S_t) \in W^{1,1}(\mathbb R)$ almost surely), with
\[
\E\!\left[
\left\|\partial_a\pi_t(\cdot\mid S_t)\right\|_{L^1(\mathbb R)}
\right]<\infty.
\]
Since $\pi_{t,\varepsilon}(a\mid s)=\pi_t(a-\varepsilon\mid s)$, we have
\[
(L_t^*\pi_t)(a\mid s)
=
-\partial_a\pi_t(a\mid s),
\]
The first-order support holds because the derivative vanishes
almost everywhere where the density is zero.
\end{revisionblock}
The backward operator is
$(L_t q)(s,a)
=
\partial_a q(s,a);$
Therefore, by Theorem~\ref{thm:gPGT}, the IPE for uniformly shifting actions upwards is
\[
\Theta
=
\sum_{t=1}^T
\gamma^{t-1}
\E\!\left[
\partial_a q_t(S_t,A_t)
\right],
\]
i.e., the IPE is the discounted sum of expected future action-gradients.
\end{example}

\begin{example}[\textbf{Adding noise}]
\label{ex:diffusion}
Again with continuous-valued actions $\mathcal A_t=\mathbb R^d$,
we might also be interested in the welfare effects of increasing uncertainty of actions. Formally, we accomplish this by setting up a perturbation path as in Definition \ref{def:operator_generated} with a diffusion process, i.e., with $M_{t,\varepsilon}(\cdot\mid a,s)$ representing the transition law at time $\varepsilon$ of the diffusion
initialized at $a$:
\[
dA_u
=
\Sigma_t(s,A_u)^{1/2}\,dW_u,
\qquad
A_0=a.
\]
\begin{revisionblock}
On the domain of its forward generator, the density derivative is
\[
(L_t^*\pi_t)(a\mid s)
=
\frac12\sum_{i,j=1}^d
\partial_{a_i a_j}
\left\{
\Sigma_{t,ij}(s,a)\pi_t(a\mid s)
\right\}.
\]
First-order support requires this derivative to vanish almost everywhere where $\pi_t(a\mid s)=0$.
A sufficient condition for this is that $\Sigma_t(s,a)=\Sigma_t(s)$ with bounded covariance norm,
and that $\pi_t(\cdot \mid s)$ admit a twice-differentiable extension to $\mathbb R^d$ satisfying
\[
\pi_t(\cdot\mid s)\in W^{2,1}(\mathbb R^d),
\qquad
\E\!\left[\|\nabla_a^2\pi_t(\cdot\mid S_t)\|_{L^1}\right]<\infty.
\]
\end{revisionblock}
The backward operator is the infinitesimal generator
$(L_t q)(s,a)
=
\frac12
\operatorname{tr}
\!\left\{
\Sigma_t(s,a)\nabla_a^2 q(s,a)
\right\};$
Therefore, by Theorem~\ref{thm:gPGT}, the IPE is
\[
\Theta
=
\frac12 \sum_{t=1}^T
\gamma^{t-1}
\E\!\left[
\operatorname{tr}
\!\left\{
\Sigma_t(S_t,A_t)\nabla_a^2 q_t(S_t,A_t)
\right\}
\right],
\]
i.e., the IPE depends on the curvature of the $q$-functions.
\end{example}

\begin{revisionblock}
\begin{remark}[Effect normalization]
\label{rem:policy_scale}
The IPE considered here is essentially a raw policy gradient along the convolutional
perturbation path. Much of the literature on related infinitesimal interventions
has instead focused on marginal policy effects (MPEs) that effectively normalize the IPE
by a measure of intervention strength. For example, \citet{carneiro2010evaluating}
normalizes marginal welfare changes by the induced net change in treatment participation;
while \citet{ghosh2025non} normalize the IPE induced by a threshold change in a dynamic
setting by the expected increase in the per-person (discounted) threshold crossings.
As discussed in both \citet{carneiro2010evaluating} and \citet{ghosh2025non}, this
type of scaling has the advantage that---under further structural assumptions---MPEs
can be interpreted as treatment effects for those affected by the policy change.
Here, we focus on IPEs for simplicity; we note, however, that our methods could
be directly extended to estimate any scaled effects defined as ratios of IPEs.
\end{remark}
\end{revisionblock}

\section{Doubly Robust Estimation}
\label{sec:estimator_construction}

Having established our identification result for the IPE, our next task is to develop
practical non-parametric methods for estimation and inference. A
first, natural idea is to estimate the $q$-functions via machine learning or
some other non-parametric approach, and then plug resulting estimates into the
identification results. By Theorem~\ref{thm:gPGT}, the IPE decomposes as
\begin{equation}
\Theta
=
\sum_{t=1}^T G_t(q_t^\star),
\qquad
G_t(q)
:=
\gamma^{t-1}\E\!\left[
(L_t q)(S_t,A_t)
\right],
\label{eq:IPE_stage_decomp}
\end{equation}
where, to distinguish oracle nuisance functions from their estimates, we write
$q_t^\star$ for the true $q$-function and
$H_t^\star$ for the true stagewise score appearing in Theorem~\ref{thm:gPGT}.

Thus, given estimates $\{\widehat q_t\}_{t=1}^T$ of the true action-value
functions $\{q_t^\star\}_{t=1}^T$, we obtain a natural plug-in estimator
\begin{equation}
\widehat\Theta^{(\mathrm{Direct})}
:=
\sum_{t=1}^T
\gamma^{t-1}\E_n\!\left[
(L_t \widehat q_t)(S_t,A_t)
\right],
\label{eq:direct_plugin_def}
\end{equation}
where $\E_n[\cdot]$ denotes the empirical average over the observed sample.
Under sequential unconfoundedness, the oracle $q$-function from
\eqref{eq:def_q} is identified by the on-policy regression relation
\begin{equation}
q_t^\star(s,a)
=
\E\!\left[
\Gamma_t
\,\middle|\,
S_t=s,\ A_t=a
\right].
\label{eq:q_reg}
\end{equation}
Thus the direct estimator only requires nonparametric regression of the
discounted cumulative outcome on $(S_t,A_t)$.
However, while it is easy to implement, the direct estimator relies heavily on outcome
modeling. Estimating $q_t^\star$ requires predicting a long-horizon discounted
outcome, and any first-order error in $\widehat q_t$ enters into the final
estimator through the operator term $(L_t\widehat q_t)(S_t,A_t)$
\citep{chernozhukov2018double}.

Now, the forward representation given in
Theorem~\ref{thm:gPGT} also suggests a weighted estimator.
Given estimates $\widehat H_t$ of the oracle Riesz-representer-like
weights $H_t^\star$, we can plug into \eqref{eq:gPGT} to form the Sequential
Riesz Weighted (SRW) estimator
\begin{equation}
\widehat\Theta^{(\mathrm{SRW})}
:=
\sum_{t=1}^T
\gamma^{t-1}\E_n\!\left[
\widehat H_t(S_t,A_t)\,\Gamma_t
\right].
\label{eq:SRW_def_unified}
\end{equation}
$\widehat\Theta^{(\mathrm{SRW})}$ is also generally consistent, but it can be
badly affected by estimation error for the same reason as the direct estimator---and
can also be noisy because it weights raw discounted outcomes directly.

We therefore have two basic estimators of the same target: A regression
plug-in estimator and a weighted estimator. Following the logic of AIPW estimation
\citep{robins1994estimation,chernozhukov2022locally}, we then seek to use the score-weighted
approach to debias the estimator obtained via the regression approach, or equivalently
vice versa.
For fixed candidate functions $(q,H)$, define the one-step augmented functional
\begin{equation}
\widetilde G_t(q,H)
:=
G_t(q)
+
\gamma^{t-1}\E\!\left[
H(S_t,A_t)\,\bigl\{\Gamma_t-q(S_t,A_t)\bigr\}
\right].
\label{eq:one_step_augmented_functional}
\end{equation}
The next lemma formalizes its mixed-bias property:

\begin{lemma}[Mixed-bias identity]
\label{Lem:mixed_bias}
Fix $t\in\{1,\ldots,T\}$. Suppose the stage-$t$ conditions of
Theorem~\ref{thm:gPGT} hold. In addition, assume
\[
\Gamma_t\in L^2,
\qquad
q_t^\star(S_t,A_t)\in L^2,
\qquad
H_t^\star(S_t,A_t)\in L^2.
\]
Let $q\in\mathcal Q_t$  (with $\mathcal Q_t$ as defined in Lemma \ref{lemm:Qclass}) and let $H$ be measurable with respect to
$\sigma(S_t,A_t)$ such that $H(S_t,A_t)\in L^2.$
Then we have the representation:
\begin{equation}
\widetilde G_t(q,H)-G_t(q_t^\star)
=
\gamma^{t-1}\E\!\left[
\bigl\{H_t^\star(S_t,A_t)-H(S_t,A_t)\bigr\}
\bigl\{q(S_t,A_t)-q_t^\star(S_t,A_t)\bigr\}
\right].
\label{eq:one_step_bias_identity}
\end{equation}
In particular,
$\widetilde G_t(q,H_t^\star)=G_t(q_t^\star)$ for any such $q$, and
$\widetilde G_t(q_t^\star,H)=G_t(q_t^\star)$ for any such $H$.
Equivalently, $\gamma^{t-1}H_t^\star$ is the unique Riesz representer of $G_t$.
\end{lemma}

We next define the empirical stagewise estimator by replacing
$(q,H)$ with its estimated counterpart $(\widehat q_t,\widehat H_t)$:
\begin{equation}
\widehat G_t
:=
\E_n\!\left[
\gamma^{t-1}(L_t \widehat q_t)(S_t,A_t)
+
\gamma^{t-1}\widehat H_t(S_t,A_t)
\bigl\{\Gamma_t-\widehat q_t(S_t,A_t)\bigr\}
\right].
\label{eq:stage_t_one_step_estimator}
\end{equation}
Summing over stages gives the Augmented Sequential Riesz Weighted (ASRW)
estimator:
\begin{equation}
\widehat\Theta
:=
\sum_{t=1}^T \widehat G_t
=
\widehat\Theta^{(\mathrm{Direct})}
+
\sum_{t=1}^T
\gamma^{t-1}\E_n\!\left[
\widehat H_t(S_t,A_t)\,
\bigl\{\Gamma_t-\widehat q_t(S_t,A_t)\bigr\}
\right].
\label{eq:ASRW_global_def_OS}
\end{equation}
Thus ASRW starts from the direct plug-in estimator and adds one weighted
cumulative-outcome residual at each stage. The two basic estimators appear as
special cases: setting $\widehat H_t\equiv0$ recovers the direct estimator,
while setting $\widehat q_t\equiv0$ recovers SRW.
In practice, we implement \eqref{eq:ASRW_global_def_OS} with cross-fitting so
that each scored trajectory is evaluated using nuisance estimates trained on
separate folds. The construction is summarized in
Algorithm~\ref{alg:asrw_crossfit}.

\begin{algorithm}[t]
\caption{Cross-fitted ASRW estimator}
\label{alg:asrw_crossfit}
\begin{algorithmic}[1]
\Require i.i.d. trajectories $\{Z_i\}_{i=1}^n$, number of folds $K$, discount factor $\gamma$, and known stagewise backward operators $\{L_t\}_{t=1}^T$
\State Partition $\{1,\ldots,n\}$ into disjoint scoring folds $I_1,\ldots,I_K$
\For{$k=1,\ldots,K$}
    \State Using the training sample $\{Z_i : i\notin I_k\}$, estimate nuisance functions $\{\widehat q_t^{(-k)},\widehat H_t^{(-k)}\}_{t=1}^T$
    \For{$i\in I_k$}
        \State Compute $\Gamma_{i,t} := \sum_{u=t}^T \gamma^{u-t} Y_{i,u}$ for $t=1,\ldots,T$
        \For{$t=1,\ldots,T$}
            \State Compute
            \[
            \widehat\psi_{i,t}^{(-k)}
            :=
            \gamma^{t-1}(L_t \widehat q_t^{(-k)})(S_{i,t},A_{i,t})
            +
            \gamma^{t-1}\widehat H_t^{(-k)}(S_{i,t},A_{i,t})
            \bigl\{
            \Gamma_{i,t}-\widehat q_t^{(-k)}(S_{i,t},A_{i,t})
            \bigr\}.
            \]
        \EndFor
    \EndFor
\EndFor
\State \textbf{Output:}
\[
\widehat\Theta
=
\frac{1}{n}\sum_{k=1}^K\sum_{i\in I_k}\sum_{t=1}^T
\widehat\psi_{i,t}^{(-k)}.
\]
\end{algorithmic}
\end{algorithm}

With this cross-fitted construction in place, write $(\widehat q_t,\widehat H_t)$ for the out-of-fold nuisance pair used to score a generic trajectory. Equivalently, one may view the nuisance functions as having been estimated on an auxiliary training sample and then evaluated on an independent scoring trajectory. Applying Lemma~\ref{Lem:mixed_bias} fold by fold with $q=\widehat q_t$ and $H=\widehat H_t$, and then averaging over folds gives
\begin{equation}
\E\!\left[
\widehat\Theta
\right]
-\Theta
=
\sum_{t=1}^{T}
\gamma^{t-1}\E\!\left[
\bigl\{H_t^\star(S_t,A_t)-\widehat H_t(S_t,A_t)\bigr\}
\bigl\{\widehat q_t(S_t,A_t)-q_t^\star(S_t,A_t)\bigr\}
\right].
\label{eq:mixed_bias_identity_full}
\end{equation}
Thus the bias vanishes if $\widehat q_t=q_t^\star$ for all $t$, or if
$\widehat H_t=H_t^\star$ for all $t$. More generally, the $t$-th stagewise bias
vanishes whenever either nuisance component is correct at that stage. Under
joint misspecification, the remainder is second order.

To formalize the asymptotic behavior of the feasible estimator, it is helpful to
compare it with the corresponding oracle estimator that uses the true nuisance
functions in the same form.

\begin{lemma}[CLT for the oracle estimator]
\label{lem:quasi_oracle_clt}
Suppose the conditions of Lemma~\ref{Lem:mixed_bias} hold for each
$t=1,\ldots,T$ with $q=q_t^\star$ and $H=H_t^\star$. Define
\begin{equation}
\widehat\Theta^{\star}
:=
\sum_{t=1}^T
\E_n\!\left[
\gamma^{t-1}(L_t q_t^\star)(S_t,A_t)
+
\gamma^{t-1}H_t^\star(S_t,A_t)\,
\bigl\{\Gamma_t-q_t^\star(S_t,A_t)\bigr\}
\right].
\label{eq:quasi_oracle_estimator}
\end{equation}
Suppose the per-trajectory summand in
\eqref{eq:quasi_oracle_estimator} has finite second moment. Then
$\E[\widehat\Theta^{\star}]=\Theta$ and
\begin{equation}
\begin{split}
&\sqrt n\,\bigl(\widehat\Theta^{\star}-\Theta\bigr)
\ \xrightarrow{d}\
\mathcal N(0,\sigma_\star^2), \\
&\sigma_\star^2
=
\Var\!\left[
\sum_{t=1}^T
\gamma^{t-1}\left\{
(L_t q_t^\star)(S_t,A_t)
+
H_t^\star(S_t,A_t)\,
\left(\Gamma_t-q_t^\star(S_t,A_t)\right )
\right\}
\right].
\end{split}
\end{equation}
\end{lemma}

The next theorem shows that the feasible cross-fitted estimator is first-order
equivalent to the oracle estimator whenever the value-side and score-side
nuisance errors are jointly small enough.

\begin{theorem}[Asymptotic normality of the ASRW estimator]
\label{thm:dr_clt}
Suppose the conditions of Lemma~\ref{lem:quasi_oracle_clt} hold, and construct
$\widehat\Theta$ using Algorithm~\ref{alg:asrw_crossfit}. 
On each scoring fold, assume $\widehat q_t\in\mathcal Q_t$ and that the conditions of
Lemma~\ref{Lem:mixed_bias} hold conditionally on the training folds with $q=\widehat q_t$ and $H=\widehat H_t$. Assume outcomes are uniformly bounded,
$|Y_t|\le Y_{\max}$ almost surely for all $t$, and let
$\Gamma_{\max}
:=
Y_{\max}\sum_{k=0}^{T-1}\gamma^k$.
Assume the regression nuisance is uniformly bounded,
$\sup_{a\in\mathcal A_t}
\bigl|\widehat q_t(S_t,a)\bigr|
\le
\Gamma_{\max}$.
\begin{revisionblock}
For fixed $T$, assume further that, on each fold and for every $t$,
\[
\|\widehat q_t-q_t^\star\|_2=o_p(1),
\qquad
\|\widehat H_t-H_t^\star\|_2=o_p(1),
\qquad
\|L_t(\widehat q_t-q_t^\star)\|_2=o_p(1),
\]
and
\[
\|\widehat q_t-q_t^\star\|_2
\|\widehat H_t-H_t^\star\|_2
=o_p(n^{-1/2}).
\]
\end{revisionblock}
Then the feasible estimator is first-order equivalent to the oracle estimator,
$\sqrt n\,\bigl(\widehat\Theta-\widehat\Theta^\star\bigr)\to_p 0$,
and so
\begin{equation}
\sqrt n\,\bigl(\widehat\Theta-\Theta\bigr)
\Rightarrow
\mathcal N(0,\sigma_\star^2),
\end{equation}
where $\sigma_\star^2$ is the oracle variance from
Lemma~\ref{lem:quasi_oracle_clt}.
\end{theorem}

Theorem~\ref{thm:dr_clt} formalizes the double-robustness claim for the
cross-fitted ASRW estimator. The direct estimator is sensitive to errors in the
value nuisance, while the score-weighted estimator is sensitive to errors in the
score nuisance. Their augmentation yields an estimating equation whose bias is
second order. As a result, for fixed time horizon $T$, ASRW is first-order equivalent to the oracle estimator whenever the value-side and score-side errors are jointly small enough, with the direct plug-in term controlled through the operator-side error $\|L_t(\widehat q_t-q_t^\star)\|_{2}.$
\revision{Finally, we close our formal analysis by verifying that our proposed approach
admits an asymptotic variance that scales polynomially with the effective horizon.}

\begin{revisionblock}
\begin{proposition}[Horizon dependence of the oracle variance]
\label{prop:horizon_variance}
Suppose the conditions of Lemma~\ref{lem:quasi_oracle_clt} hold for each horizon $T$, with the outcome bound and $\Gamma_{\max}$ as in Theorem~\ref{thm:dr_clt}.
Assume further that
\[
\|H_t^\star\|_2\le C_H,
\qquad
\|L_tq_t^\star\|_2\le C_L\Gamma_{\max},
\]
for constants $C_H$ and $C_L$ independent of $t$ and $T$. Write $\sigma_{\star,T}^2$ for the oracle variance in Lemma~\ref{lem:quasi_oracle_clt}.
Then
\[
\sigma_{\star,T}^2
\le
Y_{\max}^2(C_L+2C_H)^2
\left(\sum_{j=0}^{T-1}\gamma^j\right)^3.
\]
Consequently,
\[
\Var(\widehat\Theta^\star)
=
\frac{\sigma_{\star,T}^2}{n}
\le
\frac{Y_{\max}^2(C_L+2C_H)^2}{n}
\begin{cases}
T^3, & \gamma=1,\\[1mm]
(1-\gamma)^{-3}, & 0<\gamma<1.
\end{cases}
\]
where $\widehat\Theta^\star$ is defined in \eqref{eq:quasi_oracle_estimator}.
\end{proposition}

In other words, under uniform moment and operator bounds, our target asymptotic variance grows at most cubically without discounting and is bounded uniformly in $T$ for fixed $\gamma<1$. 
The cubic order is consistent with variance bounds for on-policy policy-gradient estimators
\citep[Appendix D]{papini2022smoothing}.
In contrast, for general off-policy evaluation with non-infinitesimal policy changes, the efficient variance will instead grow exponentially in the horizon $T$ without strong assumptions \citep[Theorem 5]{kallus2020double}.

We also note that, for the feasible estimator, nuisance estimation introduces another source of horizon dependence, since the dimension of $S_t$ increases with $t$.
Maintaining the rates in Theorem~\ref{thm:dr_clt} therefore requires control of nuisance complexity.
In practice, regularized regression and Riesz estimation provide useful tools for estimating these high-dimensional nuisance functions under suitable sparsity and regularity conditions
\citep{bradic2024high,chernozhukov2022automatic}.

\end{revisionblock}


\section{Experiments}
\label{sec:experiments}
\begingroup
\clubpenalty=10000
\widowpenalty=10000
We evaluate our proposed approach in two experiments set within the context of Example \ref{ex:location_shift}, i.e., where our proposed intervention is to infinitesimally increase the value of a continuous action. Our first experiment involves a simple auto-regressive process,
but with hidden states so basic time-series regression cannot be applied.
Our second experiment is motivated by a pricing application where demand responds to both current and past prices, in that consumers use past prices to form a (latent) reference level to which they then compare current prices.

Following Example~\ref{ex:location_shift}, the basic location shift has backward operator
\[
(L_t q)(s,a)=\partial_a q(s,a).
\]
\begin{revisionblock}
In the pricing experiment, we consider an intervention that performs a location shift with clipping, giving
\[
(L_t q)(s,a)
=
\mathbf 1\{p_{\min}<a<p_{\max}\}\partial_a q(s,a).
\]
\end{revisionblock}
Using the corresponding operator in each design,
$G_t(q)
=
\gamma^{t-1}\E\!\left[(L_t q)(S_t,A_t)\right],$
and the direct estimator is
\begin{equation*}
\widehat\Theta^{(\mathrm{Direct})}
=
\sum_{t=1}^T
\E_n\!\left[
\gamma^{t-1}(L_t\widehat q_t)(S_t,A_t)
\right].
\end{equation*}
We estimate the $q_t$ functions via direct non-parametric regression
on the full observed history and current action $(S_t,A_t)$
using feedforward neural networks.

Meanwhile, we estimate the weighting functions $\{\widehat H_t\}_{t=1}^T$
using the variational characterization of
\citet{chernozhukov2022automatic} for the linear functional
\[
\Psi_t(q)
:=
\E\!\left[(L_tq)(S_t,A_t)\right].
\]
Accordingly, on each training split we
solve the empirical variational problem
\begin{equation}
\widehat H_t
\in
\underset{h\in\mathcal H_t}{\operatorname{argmin}}
\left\{ \E_n\!\left[h(S_t,A_t)^2 - 2(L_t h)(S_t,A_t) \right] \right\}.
\label{eq:weight_learner}
\end{equation}
For the Gaussian benchmark, $(L_t h)(s,a)=\partial_a h(s,a)$.
For pricing, $(L_t h)(s,a)
=\mathbf1\{p_{\min}<a<p_{\max}\} \, \partial_a h(s,a)$.
We fit the optimization over the class $\mathcal H_t$ of feedforward neural networks
with SiLU activation functions;
network architectures and further implementation details are
provided in Online Appendix~B.2.

We report results using three estimators from
Section~\ref{sec:estimator_construction}: Direct, SRW and ASRW.
The direct estimator is obtained by plugging in $q$-function estimates into
\eqref{eq:direct_plugin_def} and SRW is obtained via plug-in with \eqref{eq:SRW_def_unified}.
The ASRW estimator is implemented as in Algorithm~\ref{alg:asrw_crossfit}.
In both experiments, we use $K$-fold cross-fitting with $K=5$.
\begin{revisionblock}
For each estimator, write
$\widehat\Theta=n^{-1}\sum_{i=1}^n\widehat\psi_i$,
where $\widehat\psi_i$ denotes its full-trajectory contribution,
evaluated using out-of-fold nuisance estimates.
Motivated by Theorem \ref{thm:dr_clt}, we compute
\begin{equation}
\label{eq:simulation_se}
\widehat{\operatorname{SE}}
=
\left\{
\frac{1}{n(n-1)}
\sum_{i=1}^n
(\widehat\psi_i-\widehat\Theta)^2
\right\}^{1/2}
\end{equation}
and report nominal $95\%$ intervals
$\widehat\Theta\pm1.96 \, \widehat{\operatorname{SE}}$.
Coverage for Direct and SRW is reported for comparison using the same interval construction.
(We emphasize that the resulting coverage numbers for the direct and SRW methods are provided for comparison only
and, unlike for ASRW, are generally not expected to be asymptotically valid \citep{chernozhukov2018double}.)
In all cases, we approximate the IPE by a forward finite difference at $\varepsilon=0$.
Online Appendix~B.2 reports further experiment details. 
\end{revisionblock}

\begin{revisionblock}
Finally, we also report results for an infeasible benchmark, namely the
ASRW estimator but with oracle weighting functions. Specifically, we replace
$\widehat H_t$ by $H_t^\star$ while retaining the same fitted $\widehat q_t$ and
cross-fitting folds. This infeasible benchmark helps us assess whether we are
in the setting covered by Theorem \ref{thm:dr_clt} since, under the assumed
accuracy bounds for the nuisance estimators, the ASRW and oracle ASRW estimators
should be first-order equivalent. In our first experiment we use a Gaussian baseline policy
$A_t\mid S_t\sim\mathcal N(m_t(S_t),v_t(S_t))$,
and so the oracle location-shift score is
\begin{equation}
H_t^\star(S_t,A_t)
=
\left.
\partial_\varepsilon
\log \pi_{t,\varepsilon}(A_t\mid S_t)
\right|_{\varepsilon=0+}
=
\frac{A_t-m_t(S_t)}{v_t(S_t)}.
\end{equation}
For pricing, the clipped-Gaussian score has interior and endpoint terms, given in Online Appendix~B.1.
We use these formulas to form the $H_t^\star$ weights for the oracle ASRW estimator; however, we emphasize
that these formulas (i.e., effectively knowledge of the distributional form of the baseline policy)
are not used for the 3 feasible estimators. SRW and ASRW learn the weights
non-parametrically through \eqref{eq:weight_learner}.
\end{revisionblock}

\subsection{A first example}
\label{sec:simulation}

We begin with a benchmark DGP designed to isolate the statistical difficulty created by downstream effects in a continuous-action setting when the analyst does not observe the full dynamic state. 
The simulator is a partially observed
Markov model: There is a latent regime $U_t\in\{-1,+1\}$, the full system is Markov in $(X_t,U_t)$, but the analyst only observes the history
\[
S_t=(X_1,A_1,Y_1,\ldots,X_{t-1},A_{t-1},Y_{t-1},X_t).
\]
Because actions depend only on $S_t$, sequential unconfoundedness is satisfied. At the same time, because $U_t$ is unobserved, the analyst does not
observe a Markov state sufficient for Bellman recursion. Consequently, classical OPE methods that rely on fitting an observed-state MDP would be misspecified.

We fix a horizon $T\in\mathbb N$ and discount factor $\gamma=0.99$, and let $X_t=(X_{t,1},\ldots,X_{t,p})^\top\in\mathbb R^p$ denote the observed covariates at time $t$. Define the loading vector
$\ell=(\ell_1,\ldots,\ell_p)^\top\in\mathbb R^p$ by
\[
\ell_j
=
\frac{j^{-1/2}}{\left(\sum_{k=1}^p k^{-1}\right)^{1/2}},
\qquad
j=1,\ldots,p.
\]
Initially,
\[
\Pr(U_1=1)=\Pr(U_1=-1)=\frac12,
\qquad
X_1 = 0.75\,U_1\,\ell + 0.60\,\zeta_1,
\qquad
\zeta_1\sim \mathcal N(0,I_p).
\]
For $t\ge 1$, let $\bar X_t:=p^{-1}\sum_{j=1}^p X_{t,j}$ and set $A_0\equiv 0$. Under the baseline behavior policy,
\begin{equation}
A_t \mid S_t \sim \mathcal N\!\bigl(m_t(S_t),\,\sigma_A^2\bigr),
\qquad
m_t(S_t)=0.10+0.35\,\bar X_t+0.15\,A_{t-1},
\label{eq:sim_behavior_policy}
\end{equation}
with $\sigma_A=1$.
Period rewards depend on both the observed state and the latent regime:
\begin{equation}
Y_t
=
\operatorname{clip}\!\left(
1+\bar X_t+1.25\,A_t+0.75\,U_t+0.50\,A_tU_t+\xi_t,
[-10,10]\right),
\label{eq:sim_reward}
\end{equation}
where $\xi_t\stackrel{\mathrm{iid}}{\sim}\mathcal N(0,1)$.
For $t<T$, the latent regime evolves according to
\begin{equation}
\Pr(U_{t+1}=1\mid X_t,U_t)
=
\operatorname{logit}^{-1}\!\bigl(U_t+0.80\,\bar X_t\bigr),
\label{eq:sim_hidden_transition}
\end{equation}
while the observed covariates follow
\begin{equation}
X_{t+1,j}
=
0.65\,X_{t,j}+0.20\,A_t+0.85\,U_t\,\ell_j+\varepsilon_{t+1,j},
\label{eq:sim_state_transition}
\end{equation}
where
$\varepsilon_{t+1,j}\stackrel{\mathrm{iid}}{\sim}\mathcal N(0,1)$
for $j=1,\ldots,p$.
Thus $U_t$ affects both current rewards and future observed covariates,
while its own transition law depends on $\bar X_t$.

We consider $n\in\{1000,2000,5000\}$, $p\in\{5,10\}$, and
$T\in\{2,3,5,10\}$, with $R=1000$ independent replications per configuration. In each replication, we simulate $n$ trajectories under the baseline policy, estimate the nuisance components using the cross-fitting pipeline above, and compute the four estimators.
We report empirical bias, root mean squared error (RMSE), and
coverage of nominal $95\%$ Wald confidence intervals.

Table~\ref{tab:IPE_results} and Figure~\ref{fig:IPE_RMSE} show a clear and stable pattern across configurations. The direct estimator exhibits substantial bias and severe undercoverage, and both worsen as the horizon grows. SRW, by contrast, suffers from high variance. ASRW delivers the smallest RMSE among the feasible estimators in almost all configurations and closely tracks ASRW (oracle score). The gains are most pronounced for larger $T$\ , where purely regression-based and weight-based methods are most fragile.

\begin{figure}[t]
\singlespacing
\captionsetup{font={normalsize},labelfont={bf}}
\centering
\includegraphics[width=0.8\textwidth]{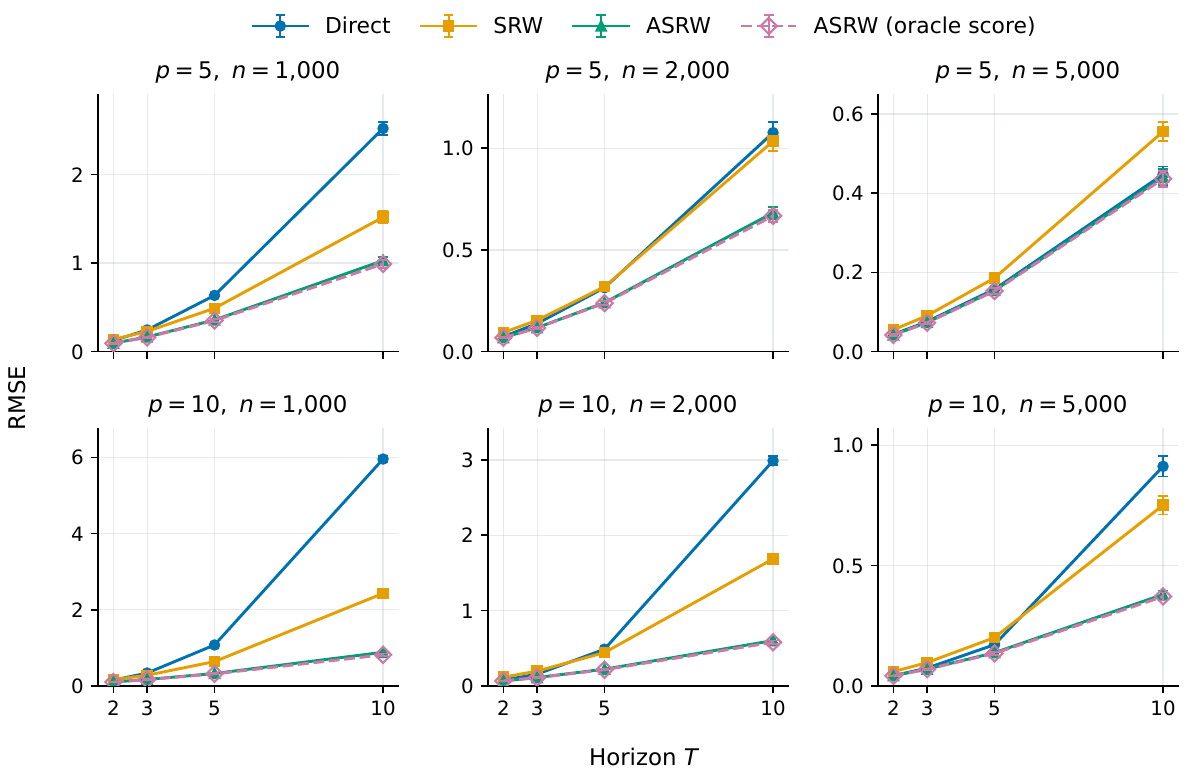}
\caption{RMSE of four estimators across all configurations.
Rows correspond to covariate dimension $p\in\{5,10\}$ and columns to sample size $n\in\{1000,2000,5000\}$.
Each panel plots RMSE against horizon $T\in\{2,3,5,10\}$.
Error bars show $\pm1.96$ Monte Carlo standard errors.}
\label{fig:IPE_RMSE}
\end{figure}

\begin{table}[p]
\singlespacing
\captionsetup{font={normalsize},labelfont={bf}}
\centering
\renewcommand{\arraystretch}{1.15}
\setlength{\tabcolsep}{2.5pt}
\footnotesize

\providecommand{\mc}[2]{\ensuremath{#1\,(#2)}}
\providecommand{\best}[1]{\boldsymbol{#1}}

\begin{adjustbox}{max width=\textwidth,keepaspectratio}
\begin{tabular}{cc *{4}{ccc}}
\toprule
\multicolumn{2}{c}{Design}
& \multicolumn{3}{c}{Direct}
& \multicolumn{3}{c}{SRW}
& \multicolumn{3}{c}{ASRW}
& \multicolumn{3}{c}{ASRW (oracle score)} \\
\cmidrule(lr){1-2}
\cmidrule(lr){3-5}
\cmidrule(lr){6-8}
\cmidrule(lr){9-11}
\cmidrule(lr){12-14}
$n$ & $p$
& Bias & RMSE & Covg
& Bias & RMSE & Covg
& Bias & RMSE & Covg
& Bias & RMSE & Covg \\
\midrule
\multicolumn{14}{c}{$T=2$} \\
\midrule

1000 & 5
& \mc{-0.05}{<1} & \mc{0.12}{<1} & $0.194$
& \mc{-0.06}{<1} & \mc{0.14}{<1} & $0.984$
& \mc{0.00}{<1} & \mc{\best{0.09}}{<1} & $0.945$
& \mc{0.00}{<1} & \mc{0.09}{<1} & $0.949$ \\

1000 & 10
& \mc{-0.08}{<1} & \mc{0.15}{<1} & $0.127$
& \mc{-0.09}{<1} & \mc{0.16}{<1} & $0.964$
& \mc{0.00}{<1} & \mc{\best{0.10}}{<1} & $0.934$
& \mc{0.00}{<1} & \mc{0.10}{<1} & $0.943$ \\

\cdashline{1-14}

2000 & 5
& \mc{-0.01}{<1} & \mc{0.07}{<1} & $0.248$
& \mc{-0.03}{<1} & \mc{0.09}{<1} & $0.994$
& \mc{0.00}{<1} & \mc{\best{0.07}}{<1} & $0.951$
& \mc{-0.00}{<1} & \mc{0.07}{<1} & $0.945$ \\

2000 & 10
& \mc{-0.01}{<1} & \mc{0.08}{<1} & $0.230$
& \mc{-0.06}{<1} & \mc{0.11}{<1} & $0.955$
& \mc{-0.00}{<1} & \mc{\best{0.07}}{<1} & $0.943$
& \mc{-0.00}{<1} & \mc{0.07}{<1} & $0.948$ \\

\cdashline{1-14}

5000 & 5
& \mc{0.00}{<1} & \mc{0.04}{<1} & $0.344$
& \mc{0.01}{<1} & \mc{0.06}{<1} & $0.997$
& \mc{0.00}{<1} & \mc{\best{0.04}}{<1} & $0.952$
& \mc{0.00}{<1} & \mc{0.04}{<1} & $0.948$ \\

5000 & 10
& \mc{0.00}{<1} & \mc{\best{0.04}}{<1} & $0.324$
& \mc{-0.01}{<1} & \mc{0.06}{<1} & $0.992$
& \mc{0.00}{<1} & \mc{0.04}{<1} & $0.951$
& \mc{-0.00}{<1} & \mc{0.04}{<1} & $0.955$ \\

\midrule
\multicolumn{14}{c}{$T=3$} \\
\midrule

1000 & 5
& \mc{-0.15}{1} & \mc{0.24}{1} & $0.129$
& \mc{-0.11}{1} & \mc{0.23}{<1} & $0.991$
& \mc{0.01}{1} & \mc{\best{0.16}}{<1} & $0.944$
& \mc{0.00}{1} & \mc{0.16}{<1} & $0.948$ \\

1000 & 10
& \mc{-0.26}{1} & \mc{0.34}{1} & $0.051$
& \mc{-0.18}{1} & \mc{0.28}{1} & $0.951$
& \mc{0.01}{1} & \mc{\best{0.17}}{<1} & $0.934$
& \mc{-0.00}{1} & \mc{0.16}{<1} & $0.942$ \\

\cdashline{1-14}

2000 & 5
& \mc{-0.04}{<1} & \mc{0.14}{<1} & $0.173$
& \mc{-0.05}{<1} & \mc{0.15}{<1} & $0.991$
& \mc{0.00}{<1} & \mc{\best{0.12}}{<1} & $0.954$
& \mc{0.00}{<1} & \mc{0.12}{<1} & $0.955$ \\

2000 & 10
& \mc{-0.07}{<1} & \mc{0.16}{<1} & $0.147$
& \mc{-0.12}{1} & \mc{0.20}{<1} & $0.955$
& \mc{0.00}{<1} & \mc{\best{0.11}}{<1} & $0.958$
& \mc{-0.00}{<1} & \mc{0.11}{<1} & $0.953$ \\

\cdashline{1-14}

5000 & 5
& \mc{0.00}{<1} & \mc{0.08}{<1} & $0.255$
& \mc{0.01}{<1} & \mc{0.09}{<1} & $0.999$
& \mc{0.00}{<1} & \mc{\best{0.07}}{<1} & $0.938$
& \mc{0.00}{<1} & \mc{0.07}{<1} & $0.948$ \\

5000 & 10
& \mc{-0.01}{<1} & \mc{0.07}{<1} & $0.233$
& \mc{-0.02}{<1} & \mc{0.10}{<1} & $0.988$
& \mc{0.00}{<1} & \mc{\best{0.07}}{<1} & $0.950$
& \mc{-0.00}{<1} & \mc{0.07}{<1} & $0.952$ \\

\midrule
\multicolumn{14}{c}{$T=5$} \\
\midrule

1000 & 5
& \mc{-0.49}{1} & \mc{0.63}{1} & $0.050$
& \mc{-0.26}{1} & \mc{0.49}{1} & $0.986$
& \mc{0.00}{1} & \mc{\best{0.36}}{1} & $0.949$
& \mc{-0.01}{1} & \mc{0.35}{1} & $0.947$ \\

1000 & 10
& \mc{-0.99}{1} & \mc{1.07}{1} & $0.007$
& \mc{-0.50}{1} & \mc{0.64}{1} & $0.912$
& \mc{0.01}{1} & \mc{\best{0.32}}{1} & $0.948$
& \mc{-0.01}{1} & \mc{0.31}{1} & $0.951$ \\

\cdashline{1-14}

2000 & 5
& \mc{-0.16}{1} & \mc{0.31}{1} & $0.103$
& \mc{-0.14}{1} & \mc{0.32}{1} & $0.997$
& \mc{0.00}{1} & \mc{\best{0.24}}{1} & $0.943$
& \mc{-0.00}{1} & \mc{0.24}{1} & $0.948$ \\

2000 & 10
& \mc{-0.40}{1} & \mc{0.49}{1} & $0.055$
& \mc{-0.34}{1} & \mc{0.44}{1} & $0.923$
& \mc{0.01}{1} & \mc{\best{0.22}}{1} & $0.941$
& \mc{-0.00}{1} & \mc{0.21}{<1} & $0.952$ \\

\cdashline{1-14}

5000 & 5
& \mc{-0.00}{1} & \mc{0.16}{<1} & $0.175$
& \mc{0.02}{1} & \mc{0.19}{<1} & $0.999$
& \mc{0.01}{1} & \mc{\best{0.15}}{<1} & $0.942$
& \mc{0.01}{1} & \mc{0.15}{<1} & $0.946$ \\

5000 & 10
& \mc{-0.08}{1} & \mc{0.17}{<1} & $0.122$
& \mc{-0.09}{1} & \mc{0.20}{<1} & $0.984$
& \mc{-0.00}{1} & \mc{\best{0.14}}{<1} & $0.959$
& \mc{-0.01}{1} & \mc{0.13}{<1} & $0.964$ \\

\midrule
\multicolumn{14}{c}{$T=10$} \\
\midrule

1000 & 5
& \mc{-2.31}{4} & \mc{2.52}{4} & $0.013$
& \mc{-1.01}{4} & \mc{1.52}{3} & $0.969$
& \mc{0.04}{4} & \mc{\best{1.03}}{2} & $0.930$
& \mc{-0.03}{4} & \mc{0.99}{2} & $0.944$ \\

1000 & 10
& \mc{-5.87}{4} & \mc{5.96}{4} & $0.000$
& \mc{-2.22}{4} & \mc{2.43}{3} & $0.679$
& \mc{-0.18}{3} & \mc{\best{0.88}}{2} & $0.932$
& \mc{-0.06}{3} & \mc{0.81}{2} & $0.960$ \\

\cdashline{1-14}

2000 & 5
& \mc{-0.80}{3} & \mc{1.08}{3} & $0.045$
& \mc{-0.67}{3} & \mc{1.03}{3} & $0.976$
& \mc{0.00}{3} & \mc{\best{0.68}}{1} & $0.951$
& \mc{-0.03}{3} & \mc{0.67}{1} & $0.952$ \\

2000 & 10
& \mc{-2.89}{3} & \mc{2.99}{3} & $0.001$
& \mc{-1.52}{3} & \mc{1.69}{3} & $0.706$
& \mc{-0.01}{3} & \mc{\best{0.60}}{1} & $0.936$
& \mc{-0.02}{3} & \mc{0.58}{1} & $0.944$ \\

\cdashline{1-14}

5000 & 5
& \mc{-0.00}{3} & \mc{0.45}{1} & $0.126$
& \mc{-0.01}{3} & \mc{0.56}{1} & $0.993$
& \mc{0.03}{3} & \mc{\best{0.44}}{1} & $0.939$
& \mc{0.02}{3} & \mc{0.44}{1} & $0.939$ \\

5000 & 10
& \mc{-0.80}{2} & \mc{0.91}{2} & $0.018$
& \mc{-0.55}{3} & \mc{0.75}{2} & $0.914$
& \mc{0.03}{2} & \mc{\best{0.38}}{1} & $0.941$
& \mc{-0.00}{2} & \mc{0.37}{1} & $0.948$ \\

\bottomrule
\end{tabular}

\end{adjustbox}

\caption{Benchmark results over $R=1000$ replications per configuration.
Parentheses contain Monte Carlo standard errors of bias and RMSE\textsuperscript{2} in units of $10^{-2}$, rounded to the nearest integer;
$<1$ denotes a positive value below $0.005$ on the original scale.
Covg is empirical coverage under the interval construction
in \eqref{eq:simulation_se}.
Boldface marks the smallest estimated RMSE among the three
feasible estimators.}
\label{tab:IPE_results}

\vspace{0.8em}
\begin{minipage}{0.98\textwidth}
\footnotesize
\textsuperscript{2}
For fixed simulation estimates $\widehat\theta_1,\ldots,\widehat\theta_M$, define
\[
g(\theta)
=
\sqrt{
M^{-1}\sum_{m=1}^M\left(\widehat\theta_m-\theta\right)^2
} \  \text{ and so } \ 
g'(\theta)
=
-
M^{-1}\sum_{m=1}^M(\widehat\theta_m-\theta)
\,/\,
g(\theta).
\]
Here $\widehat B=M^{-1}\sum_{m=1}^M(\widehat\theta_m-\theta_0)$
is the estimated bias, $\widehat R=g(\theta_0)$ is the estimated RMSE,
and $\mathrm{OSE}$ is the Monte Carlo standard error of the independently
estimated reference value. Evaluating at the oracle target $\theta_0$ gives
$g'(\theta_0)\approx-\widehat B/\widehat R$.
Therefore oracle Monte Carlo error contributes approximately
$(\widehat B/\widehat R)^2\mathrm{OSE}^2$
to the variance of the estimated RMSE.
\end{minipage}
\end{table}

\begin{revisionblock}
To assess the effect of increasing the horizon on feasible estimation,
we consider extended horizons $T\in\{2,3,5,10,15,20,30\}$,
holding $n=1000$ and $p=5$ fixed.
We consider $\gamma\in\{1,0.99\}$ and use 200 independent
replications per configuration.
Implementation details are provided in Online Appendix~B.2.
Figure~\ref{fig:feasible_horizon} shows how RMSE and empirical
standard deviation vary with $T$ for all four estimators.
Both measures increase with the horizon under both discount settings.
ASRW maintains lower empirical RMSE than Direct and SRW
throughout the horizons.
\end{revisionblock}

\begin{figure}[t]
\singlespacing
\captionsetup{font={normalsize},labelfont={bf}}
\centering
\includegraphics[width=0.8\textwidth]{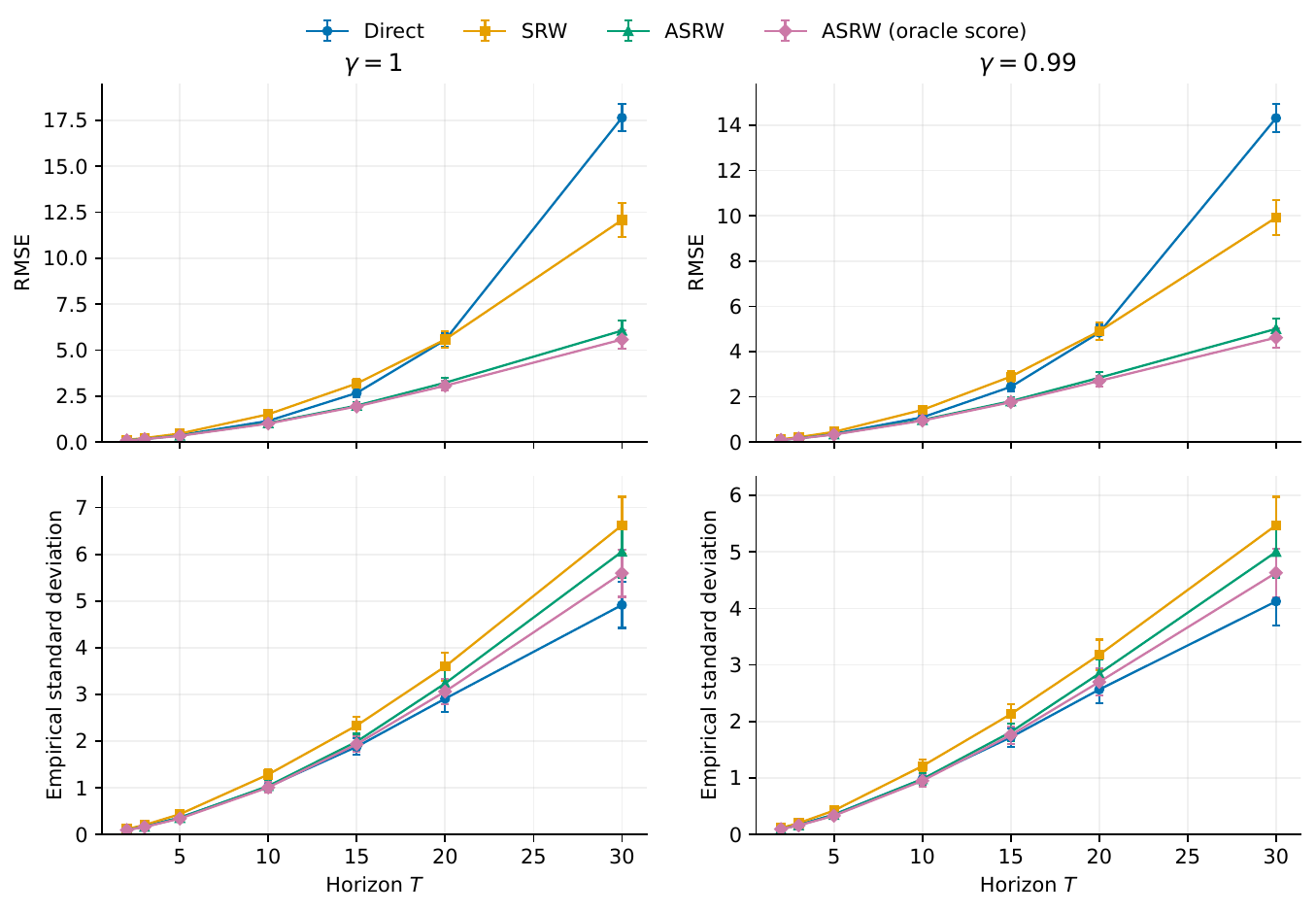}
\caption{RMSE and empirical standard deviation against horizon $T$ for four estimators.
Columns correspond to $\gamma=1$ and $\gamma=0.99$.
Error bars show $\pm1.96$ Monte Carlo standard errors.}
\label{fig:feasible_horizon}
\end{figure}

\subsection{Dynamic pricing simulator with reference effects}
\label{sec:empirical-pricing}

We now turn to a setting motivated by a pricing application in which downstream effects arise through an explicit behavioral channel. In dynamic pricing, a higher current price typically lowers contemporaneous demand but can also raise the consumer's
internal reference level, making later prices appear more attractive relative to that benchmark. This creates the same immediate-versus-downstream reward tradeoff that motivates the IPE.

We simulate a dynamic pricing problem over periods $t=1,\ldots,T$, where the platform posts a scalar price
$A_{i,t}\in[p_{\min},p_{\max}]$ for consumer $i$.
The observed history is
\[
S_{i,t}
=
(X_{i,1},A_{i,1},Y_{i,1},\ldots,
X_{i,t-1},A_{i,t-1},Y_{i,t-1},X_{i,t}).
\]
Each consumer has latent heterogeneity $U_i:=(W_i,\alpha_i)$,
where $W_i$ governs willingness to pay and $\alpha_i$ captures how quickly the consumer updates an internal reference level.
The individual's internal reference price $r^\star_{i,t}$,
which evolves over time, is also latent, as is a persistent
taste state $u_{i,t}$.
The observed state $X_{i,t}$ is a low-dimensional summary of recent prices, recent revenues, and seasonality.
Online Appendix~B.1 gives the exact state definition and
transition equations.

Conditional on the observed history, the current price, and the latent state, the purchase indicator $B_{i,t}\in\{0,1\}$ is drawn from a logistic demand model of the form
\[
\begin{aligned}
&\Pr(B_{i,t}=1\mid S_{i,t},A_{i,t},U_i,r^\star_{i,t},u_{i,t})\\
&\quad=
\sigma\!\left(
f_{\mathrm{season}}(t)
+
f_{\mathrm{direct}}(S_{i,t},A_{i,t};U_i)
+
f_{\mathrm{ref}}(r^\star_{i,t}-A_{i,t};U_i)
+u_{i,t}
\right),
\end{aligned}
\]
where $\sigma(\cdot)$ is the sigmoid function.
Here $f_{\mathrm{season}}$ captures predictable seasonal demand variation, $f_{\mathrm{direct}}$ reflects the usual direct price effect, and $f_{\mathrm{ref}}$ governs the reference effect through the gap between the consumer's internal benchmark and the posted price. 
In particular, the direct component is decreasing in the current price, while the reference component allows nonlinear effects of the reference-price gap.
Period revenue is
$Y_{i,t}:=A_{i,t}B_{i,t},$
and the discounted cumulative reward is
$\Gamma_{i,t}:=\sum_{u=t}^T\gamma^{u-t}Y_{i,u}.$
Current prices affect subsequent demand through the latent
reference-price update.
Higher prices can reduce current purchases while raising the
reference level used in later decisions.

Data are generated under a clipped-Gaussian baseline pricing policy
$\pi=(\pi_1,\ldots,\pi_T)$ that maps the current information set to
a distribution over prices.
In our implementation, the policy takes the form
\[
A^{\mathrm{raw}}_{i,t}\mid S_{i,t}
\sim
N\bigl(\mu_t(X_{i,t}),\sigma_{\mathrm{price}}^2\bigr),
\qquad
A_{i,t}:=
\operatorname{clip}\bigl(A^{\mathrm{raw}}_{i,t},
[p_{\min},p_{\max}]\bigr).
\]
The mean function $\mu_t(X_{i,t})$ is a pricing rule that incorporates inertia, recent revenue feedback, and seasonality while respecting business bounds.
The intervention shifts the latent Gaussian mean before clipping:
\[
A^{\mathrm{raw}}_{i,t,\varepsilon}\mid S_{i,t}
\sim
N\bigl(\mu_t(X_{i,t})+\varepsilon,\sigma_{\mathrm{price}}^2\bigr),
\qquad
A_{i,t,\varepsilon}:=
\operatorname{clip}\bigl(A^{\mathrm{raw}}_{i,t,\varepsilon},
[p_{\min},p_{\max}]\bigr).
\]
Although the simulator contains latent consumer heterogeneity
and a latent reference state, sequential unconfoundedness holds by construction.
The Gaussian assignment innovations are independent of the initial variables and all other innovations.
Latent states affect outcomes and transitions, while the assignment rule uses only observed history.

\begin{table}[t]
\singlespacing
\captionsetup{font={normalsize},labelfont={bf}}
\centering
\footnotesize
\begin{adjustbox}{max width=\textwidth}
\begin{tabular}{lrrrrr}
\toprule
Estimator & Bias (MCSE) & RMSE (MCSE) & Emp. SD & Mean SE & Covg\\
\midrule
Direct & $0.832\,(0.024)$ & $0.919\,(0.022)$ & 0.391 & 0.010 & 0.000\\
SRW & $0.127\,(0.061)$ & $1.319\,(0.042)$ & 1.314 & 0.692 & 0.698\\
ASRW & $0.143\,(0.024)$ & $0.406\,(0.014)$ & 0.380 & 0.398 & 0.946\\
ASRW (oracle score) & $-0.019\,(0.023)$ & $0.369\,(0.012)$ & 0.369 & 0.390 & 0.954\\
\bottomrule
\end{tabular}
\end{adjustbox}
\caption{Pricing results over $R=500$ independent replications,
with $T=8$, $n=5000$, and $\gamma=0.99$.
Parentheses contain Monte Carlo standard errors, including
independent reference-value sampling uncertainty for bias and RMSE.
Emp. SD is the standard deviation of replication estimates;
Mean SE is the mean within-replication standard error.
Covg uses \eqref{eq:simulation_se}.}
\label{tab:empirical_study}
\end{table}

\begin{figure}[t]
\singlespacing
\captionsetup{font={normalsize},labelfont={bf}}
\centering
\includegraphics[width=0.8\linewidth]{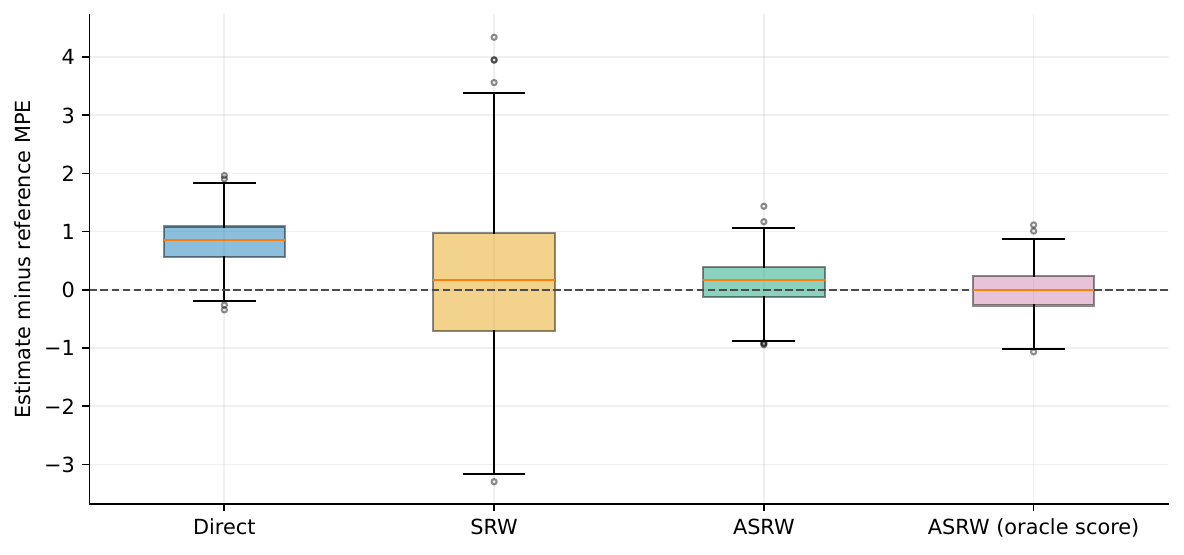}
\caption{Pricing estimates minus the numerical reference IPE over 500 independent replications.
The horizontal line is zero.
Whiskers extend to 1.5 interquartile ranges.}
\label{fig:boxplot_closed_form}
\end{figure}

Table~\ref{tab:empirical_study} reports bias, RMSE, empirical standard deviation, mean estimated standard error, and coverage.
The estimated RMSEs are $0.919$ for Direct, $1.319$ for SRW, and $0.406$ for ASRW.
ASRW has empirical coverage $0.946$.
For Direct, the mean estimated standard error is $0.010$, compared with an empirical standard deviation of $0.391$ across replications.
Figure~\ref{fig:boxplot_closed_form} shows the replication
distributions relative to the numerical reference value.
Overall, we see that the feasible ASRW outperforms baselines, and nearly matches the performance of its oracle counterpart. 
\endgroup

\section{Discussion}
\label{sec:conclusion}

In this paper, we developed a framework for observational-study causal inference in dynamic designs. We focused on estimation and inference on a class of infinitesimal policy effects, i.e., the pathwise derivative of long-run discounted welfare along a smooth, user-specified convolutional perturbation of the status-quo policy. Working under the flexible causal model of \citet{robins1986new} with sequential unconfoundedness and allowing general history dependence, we showed that this local welfare derivative admits a policy-gradient representation that decomposes the effect into stagewise contributions involving the baseline action-value functions. Building on this decomposition, we proposed a  debiasing strategy based on stagewise Riesz representers, yielding an augmented estimator that combines a regression plug-in with a single per-period correction. The resulting procedure is doubly robust and supports asymptotically valid inference under cross-fitting with flexible machine-learning nuisance estimation. 

Our key finding is that focusing on infinitesimal policy effects avoids the full off-policy evaluation problem, and enables us to derive simple non-parametric estimators whose behavior remains stable even as the decision horizon $T$ grows. In contrast, standard methods for non-parametric off-policy evaluation---including doubly-robust ones---either incur an exponential variance blowup as the horizon $T$ grows \citep{jiang2016doubly,thomas2016data}, or require the ability to observe all state variables that mediate temporal carryovers \citep{kallus2022efficiently}. Thus, in settings such as our pricing simulator that involve a moderately large $T$ and fundamentally unobservable state variables (because the states depend on private consumer beliefs), our approach enables us to make policy-relevant counterfactual inferences (i.e., learn about how long-term welfare might change if we nudge prices up or down) in environments where standard methods for generic off-policy evaluation are simply not applicable. As such, we argue that infinitesimal policy effects form a statistical target of interest and provide, in some application areas, a rigorous yet pragmatic target for observational-study causal inference in dynamic settings.


\revision{
\section{Acknowledgment}
We used GPT-6 Astra to assist with reviewing the regularity conditions and proofs in the appendix and organizing the numerical simulations and dynamic pricing experiments.
}

\bibliographystyle{plainnat}
\bibliography{ref}

@article{mehrabi2024off,
  title={Off-policy evaluation in {M}arkov decision processes under weak distributional overlap},
  author={Mehrabi, Mohammad and Wager, Stefan},
  journal={arXiv preprint arXiv:2402.08201},
  year={2024}
}

@book{levin2017markov,
  title={Markov chains and mixing times},
  author={Levin, David A and Peres, Yuval},
  volume={107},
  year={2017},
  publisher={American Mathematical Society}
}

@article{robins1994estimation,
  title={Estimation of regression coefficients when some regressors are not always observed},
  author={Robins, James M and Rotnitzky, Andrea and Zhao, Lue Ping},
  journal={Journal of the American statistical Association},
  volume={89},
  number={427},
  pages={846--866},
  year={1994},
  publisher={Taylor \& Francis}
}

@article{bodory2022evaluating,
  title={Evaluating (weighted) dynamic treatment effects by double machine learning},
  author={Bodory, Hugo and Huber, Martin and Laff{\'e}rs, Luk{\'a}{\v{s}}},
  journal={The Econometrics Journal},
  volume={25},
  number={3},
  pages={628--648},
  year={2022},
  publisher={Oxford University Press}
}

@article{bradic2024high,
  title={High-dimensional inference for dynamic treatment effects},
  author={Bradic, Jelena and Ji, Weijie and Zhang, Yuqian},
  journal={The Annals of Statistics},
  volume={52},
  number={2},
  pages={415--440},
  year={2024},
  publisher={Institute of Mathematical Statistics}
}

@article{kallus2020double,
  title={Double reinforcement learning for efficient off-policy evaluation in {M}arkov decision processes},
  author={Kallus, Nathan and Uehara, Masatoshi},
  journal={Journal of Machine Learning Research},
  volume={21},
  number={167},
  pages={1--63},
  year={2020}
}

@article{hirshberg2021augmented,
  title={Augmented minimax linear estimation},
  author={Hirshberg, David A and Wager, Stefan},
  journal={The Annals of Statistics},
  volume={49},
  number={6},
  pages={3206--3227},
  year={2021},
  publisher={Institute of Mathematical Statistics}
}

@article{dudik2014doubly,
  title   = {Doubly robust policy evaluation and optimization},
  author  = {Dud{\'i}k, Miroslav and Erhan, Dumitru and Langford, John and Li, Lihong},
  journal = {Statistical Science},
  volume  = {29},
  number  = {4},
  pages   = {485--511},
  year    = {2014},
  month   = nov,
  publisher = {Institute of Mathematical Statistics},
}

@article{chernozhukov2022automatic,
  title={Automatic debiased machine learning of causal and structural effects},
  author={Chernozhukov, Victor and Newey, Whitney K and Singh, Rahul},
  journal={Econometrica},
  volume={90},
  number={3},
  pages={967--1027},
  year={2022},
  publisher={Wiley Online Library}
}

@article{marbach2001simulation,
  title={Simulation-based optimization of Markov reward processes},
  author={Marbach, Peter and Tsitsiklis, John N},
  journal={IEEE Transactions on Automatic Control},
  volume={46},
  number={2},
  pages={191--209},
  year={2001},
  publisher={IEEE}
}

@article{murphy2003optimal,
  title={Optimal dynamic treatment regimes},
  author={Murphy, Susan A},
  journal={Journal of the Royal Statistical Society Series B: Statistical Methodology},
  volume={65},
  number={2},
  pages={331--355},
  year={2003},
  publisher={Oxford University Press}
}

@article{chakraborty2014dynamic,
  title={Dynamic treatment regimes},
  author={Chakraborty, Bibhas and Murphy, Susan A},
  journal={Annual review of statistics and its application},
  volume={1},
  number={1},
  pages={447--464},
  year={2014},
  publisher={Annual Reviews}
}

@article{robins1986new,
  title={A new approach to causal inference in mortality studies with a sustained exposure period—application to control of the healthy worker survivor effect},
  author={Robins, James},
  journal={Mathematical modelling},
  volume={7},
  number={9-12},
  pages={1393--1512},
  year={1986},
  publisher={Elsevier}
}

@article{ertefaie2018constructing,
  title={Constructing dynamic treatment regimes over indefinite time horizons},
  author={Ertefaie, Ashkan and Strawderman, Robert L},
  journal={Biometrika},
  volume={105},
  number={4},
  pages={963--977},
  year={2018},
  publisher={Oxford University Press}
}

@inproceedings{jiang2016doubly,
  title={Doubly robust off-policy value evaluation for reinforcement learning},
  author={Jiang, Nan and Li, Lihong},
  booktitle={International conference on machine learning},
  pages={652--661},
  year={2016},
  organization={PMLR}
}

@inproceedings{robins2004optimal,
  title={Optimal structural nested models for optimal sequential decisions},
  author={Robins, James M},
  booktitle={Proceedings of the Second Seattle Symposium in Biostatistics: analysis of correlated data},
  pages={189--326},
  year={2004},
  organization={Springer}
}

@article{kennedy2019nonparametric,
  title={Nonparametric causal effects based on incremental propensity score interventions},
  author={Kennedy, Edward H},
  journal={Journal of the American Statistical Association},
  volume={114},
  number={526},
  pages={645--656},
  year={2019},
  publisher={Taylor \& Francis}
}

@article{bojinov2023design,
  title={Design and analysis of switchback experiments},
  author={Bojinov, Iavor and Simchi-Levi, David and Zhao, Jinglong},
  journal={Management Science},
  volume={69},
  number={7},
  pages={3759--3777},
  year={2023},
  publisher={INFORMS}
}

@article{hu2022switchback,
  title={Switchback experiments under geometric mixing},
  author={Hu, Yuchen and Wager, Stefan},
  journal={Management Science},
  volume={forthcoming},
  year={2026}
}

@article{glynn2020adaptive,
  title={Adaptive experimental design with temporal interference: A maximum likelihood approach},
  author={Glynn, Peter W and Johari, Ramesh and Rasouli, Mohammad},
  journal={Advances in Neural Information Processing Systems},
  volume={33},
  pages={15054--15064},
  year={2020}
}

@article{farias2022markovian,
  title={Markovian interference in experiments},
  author={Farias, Vivek and Li, Andrew and Peng, Tianyi and Zheng, Andrew},
  journal={Advances in Neural Information Processing Systems},
  volume={35},
  pages={535--549},
  year={2022}
}

@article{johari2025estimation,
  title={Estimation of Treatment Effects Under Nonstationarity via the Truncated Policy Gradient Estimator},
  author={Johari, Ramesh and Peng, Tianyi and Xing, Wenqian},
  journal={arXiv preprint arXiv:2506.05308},
  year={2025}
}

@book{tsiatis2019dynamic,
  title={Dynamic treatment regimes: Statistical methods for precision medicine},
  author={Tsiatis, Anastasios A and Davidian, Marie and Holloway, Shannon T and Laber, Eric B},
  year={2019},
  publisher={Chapman and Hall/CRC}
}

@article{liao2022batch,
  title   = {Batch policy learning in average reward Markov decision processes},
  author  = {Liao, Peng and Qi, Zhengling and Wan, Runzhe and Klasnja, Predrag and Murphy, Susan A.},
  journal = {The Annals of Statistics},
  volume  = {50},
  number  = {6},
  pages   = {3364--3387},
  year    = {2022}
}

@book{sutton1998reinforcement,
  title={Reinforcement learning: An introduction},
  author={Sutton, Richard S and Barto, Andrew G},
  year={1998},
  publisher={MIT press Cambridge}
}

@article{kallus2022efficiently,
  title={Efficiently breaking the curse of horizon in off-policy evaluation with double reinforcement learning},
  author={Kallus, Nathan and Uehara, Masatoshi},
  journal={Operations Research},
  volume={70},
  number={6},
  pages={3282--3302},
  year={2022},
  publisher={INFORMS}
}

@inproceedings{NIPS1999_464d828b,
 author = {Sutton, Richard S and McAllester, David and Singh, Satinder and Mansour, Yishay},
 booktitle = {Advances in Neural Information Processing Systems},
 editor = {S. Solla and T. Leen and K. M\"{u}ller},
 pages = {},
 publisher = {MIT Press},
 title = {Policy Gradient Methods for Reinforcement Learning with Function Approximation},
 volume = {12},
 year = {1999}
}

@article{williams1992simple,
  title={Simple statistical gradient-following algorithms for connectionist reinforcement learning},
  author={Williams, Ronald J},
  journal={Machine learning},
  volume={8},
  number={3},
  pages={229--256},
  year={1992},
  publisher={Springer}
}

@inproceedings{NIPS2001_4b86abe4,
 author = {Kakade, Sham M},
 booktitle = {Advances in Neural Information Processing Systems},
 editor = {T. Dietterich and S. Becker and Z. Ghahramani},
 pages = {},
 publisher = {MIT Press},
 title = {A Natural Policy Gradient},
 volume = {14},
 year = {2001}
}

@incollection{robins1997causal,
  title={Causal inference from complex longitudinal data},
  author={Robins, James M},
  booktitle={Latent variable modeling and applications to causality},
  pages={69--117},
  year={1997},
  publisher={Springer}
}

@article{wager2021experimenting,
  title={Experimenting in equilibrium},
  author={Wager, Stefan and Xu, Kuang},
  journal={Management Science},
  volume={67},
  number={11},
  pages={6694--6715},
  year={2021},
  publisher={INFORMS}
}

@article{munro2025treatment,
  title={Treatment effects in market equilibrium},
  author={Munro, Evan and Kuang, Xu and Wager, Stefan},
  journal={American Economic Review},
  volume={115},
  number={10},
  pages={3273--3321},
  year={2025},
  publisher={American Economic Association 2014 Broadway, Suite 305, Nashville, TN 37203}
}

@article{chernozhukov2022locally,
  title={Locally robust semiparametric estimation},
  author={Chernozhukov, Victor and Escanciano, Juan Carlos and Ichimura, Hidehiko and Newey, Whitney K and Robins, James M},
  journal={Econometrica},
  volume={90},
  number={4},
  pages={1501--1535},
  year={2022},
  publisher={Wiley Online Library}
}

@article{ghosh2025non,
  title={Non-parametric Causal Inference in Dynamic Thresholding Designs},
  author={Ghosh, Aditya and Wager, Stefan},
  journal={arXiv preprint arXiv:2512.15244},
  year={2025}
}

@misc{wager2024causal,
  title={Causal inference: A statistical learning approach},
  author={Wager, Stefan},
  year={2024},
  publisher={Technical report, Stanford University}
}

@book{hernan2020causal,
  title     = {Causal Inference: What If},
  author    = {Hern{\'a}n, Miguel A. and Robins, James M.},
  publisher = {Chapman \& Hall/CRC},
  address   = {Boca Raton, FL},
  year      = {2020}
}

@article{heckman2005structural,
  title   = {Structural Equations, Treatment Effects, and Econometric Policy Evaluation},
  author  = {Heckman, James J. and Vytlacil, Edward},
  journal = {Econometrica},
  year    = {2005},
  volume  = {73},
  number  = {3},
  pages   = {669--738}
}

@article{carneiro2010evaluating,
  title   = {Evaluating Marginal Policy Changes and the Average Effect of Treatment for Individuals at the Margin},
  author  = {Carneiro, Pedro and Heckman, James J. and Vytlacil, Edward},
  journal = {Econometrica},
  year    = {2010},
  volume  = {78},
  number  = {1},
  pages   = {377--394}
}

@inproceedings{thomas2016data,
  title     = {Data-Efficient Off-Policy Policy Evaluation for Reinforcement Learning},
  author    = {Thomas, Philip S. and Brunskill, Emma},
  booktitle = {Proceedings of the 33rd International Conference on Machine Learning (ICML)},
  series    = {Proceedings of Machine Learning Research},
  volume    = {48},
  pages     = {2139--2148},
  year      = {2016}
}

@article{naimi2021incremental,
  title   = {Incremental Propensity Score Effects for Time-Fixed Exposures},
  author  = {Naimi, Ashley I. and Rudolph, Jacqueline E. and Kennedy, Edward H. and Cartus, Abigail and Kirkpatrick, Sharon I. and Haas, David M. and Simhan, Hyagriv and Bodnar, Lisa M.},
  journal = {Epidemiology},
  year    = {2021},
  volume  = {32},
  number  = {2},
  pages   = {202--208}
}

@article{sarvet2023longitudinal,
  title   = {Longitudinal incremental propensity score interventions for limited resource settings},
  author  = {Sarvet, Aaron L. and Wanis, Kerollos N. and Young, Jessica G. and Hernandez-Alejandro, Roberto and Stensrud, Mats J.},
  journal = {Biometrics},
  volume  = {79},
  number  = {4},
  pages   = {3418--3430},
  year    = {2023},
  doi     = {10.1111/biom.13859}
}

@article{diaz2023lmtp,
  title   = {Nonparametric Causal Effects Based on Longitudinal Modified Treatment Policies},
  author  = {D{\'\i}az, Iv{\'a}n and Williams, Nicholas and Hoffman, Katherine L. and Schenck, Edward J.},
  journal = {Journal of the American Statistical Association},
  year    = {2023},
  volume  = {118},
  number  = {542},
  pages   = {846--857}
}

@article{haneuse2013estimation,
  title   = {Estimation of the Effect of Interventions That Modify the Received Treatment},
  author  = {Haneuse, Sebastian and Rotnitzky, Andrea},
  journal = {Statistics in Medicine},
  year    = {2013},
  volume  = {32},
  number  = {30},
  pages   = {5260--5277}
}

@article{sasaki2023prte,
  title   = {Estimation and Inference for Policy Relevant Treatment Effects},
  author  = {Sasaki, Yuya and Ura, Takuya},
  journal = {Journal of Econometrics},
  year    = {2023},
  volume  = {234},
  number  = {2},
  pages   = {394--450}
}

@article{chernozhukov2018double,
  title   = {Double/debiased machine learning for treatment and structural parameters},
  author  = {Chernozhukov, Victor and Chetverikov, Denis and Demirer, Mert and Duflo, Esther and Hansen, Christian and Newey, Whitney K. and Robins, James},
  journal = {The Econometrics Journal},
  volume  = {21},
  number  = {1},
  pages   = {C1--C68},
  year    = {2018}
}

@article{papini2022smoothing,
  title   = {Smoothing Policies and Safe Policy Gradients},
  author  = {Papini, Matteo and Pirotta, Matteo and Restelli, Marcello},
  journal = {Machine Learning},
  volume  = {111},
  pages   = {4081--4137},
  year    = {2022},
  doi     = {10.1007/s10994-022-06232-6}
}

\clearpage

\appendix
\section{Proof of technical results}

\begin{revisionblock}
\subsection{Proof of Lemma~1}
\label{app:backward_operator}
\begin{proof}
All expectations are taken under the baseline law. We write 
\[
M_{t,\varepsilon}q:=\E_{A\sim M_{t,\varepsilon}(\cdot\mid a_0,s)}
[q(s,A)]
\]
For $q,r\in\mathcal Q_t$ and $\alpha,\beta\in\mathbb R$,
linearity of $M_{t,\varepsilon}$ gives
\[
\frac{M_{t,\varepsilon}(\alpha q+\beta r)-(\alpha q+\beta r)}
{\varepsilon}
=
\alpha\frac{M_{t,\varepsilon}q-q}{\varepsilon}
+
\beta\frac{M_{t,\varepsilon}r-r}{\varepsilon}.
\]
Taking $L^1$ limits yields
$L_t(\alpha q+\beta r)=\alpha L_tq+\beta L_tr$.

The kernel-generated policy relation and boundedness of $q$ imply
\begin{align*}
&\E\!\left[
\frac{(M_{t,\varepsilon}q)(S_t,A_t)-q(S_t,A_t)}
{\varepsilon}
\,\middle|\,S_t
\right]=
\int_{\mathcal A_t}
q(S_t,a)
\frac{\pi_{t,\varepsilon}(a\mid S_t)-\pi_t(a\mid S_t)}
{\varepsilon}
\,d\lambda(a).
\end{align*}
By backward differentiability and the $L^1$ contraction property
of conditional expectation, the left side converges in $L^1$ to
$\E[(L_tq)(S_t,A_t)\mid S_t]$.
For the right side,
\begin{align*}
&\E\!\left|
\int_{\mathcal A_t}
q(S_t,a)
\left\{
\frac{\pi_{t,\varepsilon}(a\mid S_t)-\pi_t(a\mid S_t)}
{\varepsilon}
-
(L_t^*\pi_t)(a\mid S_t)
\right\}
\,d\lambda(a)
\right|\\
&\qquad\leq
\|q\|_\infty
\E\!\left[
\int_{\mathcal A_t}
\left|
\frac{\pi_{t,\varepsilon}(a\mid S_t)-\pi_t(a\mid S_t)}
{\varepsilon}
-
(L_t^*\pi_t)(a\mid S_t)
\right|
\,d\lambda(a)
\right]
\longrightarrow0.
\end{align*}
Consequently, the two $L^1$ limits agree, giving
\[
\E[(L_tq)(S_t,A_t)\mid S_t]
=
\int_{\mathcal A_t}
q(S_t,a)(L_t^*\pi_t)(a\mid S_t)\,d\lambda(a)
\qquad\text{almost surely}.
\]
\end{proof}
    
\end{revisionblock}

\subsection{Proof of Theorem~2}
\label{app:pg}

\begin{revisionblock}
\begin{proof}
Write $\E_0=\E$ for baseline expectation and
$J(\varepsilon)=\E_\varepsilon[\Gamma_1]$.
For each $t$, let $q_t\in\mathcal Q_t$ denote the bounded extension of the baseline action-value function specified in the theorem.
By Lemma~1, the policy-side $L^1$ differentiability implies integrability of the
forward derivative, i.e. for a sufficiently small $\varepsilon_0>0$,
\[
\E_0\int |L_t^*\pi_t|\,d\lambda
\leq
\frac{2}{\varepsilon_0}
+
\E_0\int\left|
\frac{\pi_{t,\varepsilon_0}-\pi_t}{\varepsilon_0}
-L_t^*\pi_t
\right|\,d\lambda
<\infty.
\]
First-order support gives $L_t^*\pi_t=H_t\pi_t$.
Since each policy integrates to one,
\begin{align*}
\E_0\left|\E_0[H_t(S_t,A_t)\mid S_t]\right|
&=
\E_0\left|\int(L_t^*\pi_t)(a\mid S_t)\,d\lambda(a)\right|\\
&\leq
\E_0\int\left|
\frac{\pi_{t,\varepsilon}-\pi_t}{\varepsilon}
-L_t^*\pi_t
\right|\,d\lambda
\longrightarrow0.
\end{align*}
The left side is independent of $\varepsilon$. Therefore,
\[
\E_0[H_t(S_t,A_t)\mid S_t]=0,
\qquad
\E_0|H_t(S_t,A_t)|
=
\E_0\int|L_t^*\pi_t|\,d\lambda
<\infty.
\]

\paragraph{Step 1. Construct the auxiliary policy.}
At every history satisfying
$\int|H_t|\pi_t\,d\lambda<\infty$ and
$\int H_t\pi_t\,d\lambda=0$, define
\[
\widetilde\pi_{t,\varepsilon}(a\mid s)
=
\frac{[1+\varepsilon H_t(s,a)]_+\pi_t(a\mid s)}
{\displaystyle
\int_{\mathcal A_t}[1+\varepsilon H_t(s,b)]_+
\pi_t(b\mid s)\,d\lambda(b)},
\qquad [x]_+=\max\{x,0\}.
\]
and set $\widetilde\pi_{t,\varepsilon}(\cdot\mid s)=\pi_t(\cdot\mid s)$ where either condition fails.

Fix a history satisfying both conditions and suppress $s$ in the following integrals. Notice that 
\[
\int(1+\varepsilon H_t)\pi_t\,d\lambda=1,
\]
and hence
\[
1\leq
\int[1+\varepsilon H_t]_+\pi_t\,d\lambda
\leq
1+\varepsilon\int|H_t|\pi_t\,d\lambda
<\infty.
\]
Thus $\widetilde\pi_{t,\varepsilon}$ is a valid probability density,
$\widetilde\pi_{t,\varepsilon}\ll\pi_t$, and
$\widetilde\pi_{t,0}=\pi_t$.

To compare the auxiliary and perturbed policies, by triangle inequality, decompose
\begin{align*}
\int|\pi_{t,\varepsilon}-\widetilde\pi_{t,\varepsilon}|\,d\lambda
\leq{}&
\underbrace{
\int|\pi_{t,\varepsilon}-(1+\varepsilon H_t)\pi_t|\,d\lambda
}_{\text{approximation error}}\\
&+
\underbrace{
\int|(1+\varepsilon H_t)\pi_t-[1+\varepsilon H_t]_+\pi_t|\,d\lambda
}_{\text{clipping correction}}\\
&+
\underbrace{
\int|[1+\varepsilon H_t]_+\pi_t-\widetilde\pi_{t,\varepsilon}|\,d\lambda
}_{\text{normalization correction}}.
\end{align*}

\emph{(I) Clipping.}
Replacing the negative part by zero gives
\[
\int|[1+\varepsilon H_t]_+\pi_t-(1+\varepsilon H_t)\pi_t|\,d\lambda
=
\int[-1-\varepsilon H_t]_+\pi_t\,d\lambda.
\]

\emph{(II) Normalization.}
Since the linear approximation integrates to one,
\[
\int[1+\varepsilon H_t]_+\pi_t\,d\lambda
=
1+\int[-1-\varepsilon H_t]_+\pi_t\,d\lambda.
\]
The normalizing denominator is at least one, so
$0\leq\widetilde\pi_{t,\varepsilon}\leq[1+\varepsilon H_t]_+\pi_t$.
Consequently,
\begin{align*}
\int|[1+\varepsilon H_t]_+\pi_t-\widetilde\pi_{t,\varepsilon}|\,d\lambda
&=
\int[1+\varepsilon H_t]_+\pi_t\,d\lambda-1\\
&=
\int[-1-\varepsilon H_t]_+\pi_t\,d\lambda.
\end{align*}
Both corrections therefore have the same $L^1$ norm.

Nonnegativity of $\pi_{t,\varepsilon}$ implies
\[
[-1-\varepsilon H_t]_+\pi_t
\leq
|\pi_{t,\varepsilon}-(1+\varepsilon H_t)\pi_t|.
\]
Clipping and normalization act on disjoint regions, giving
\begin{align*}
\int|\widetilde\pi_{t,\varepsilon}-(1+\varepsilon H_t)\pi_t|\,d\lambda
&=
2\int[-1-\varepsilon H_t]_+\pi_t\,d\lambda\\
&\leq
2\int|\pi_{t,\varepsilon}-\pi_t-\varepsilon H_t\pi_t|\,d\lambda.
\end{align*}
The triangle inequality then yields
\[
\int|\pi_{t,\varepsilon}-\widetilde\pi_{t,\varepsilon}|\,d\lambda
\leq
3\int|\pi_{t,\varepsilon}-\pi_t-\varepsilon H_t\pi_t|\,d\lambda.
\]
Taking expectation on both sides and applying policy-side
$L^1$ differentiability gives
\begin{align}
\E_0\int
|\widetilde\pi_{t,\varepsilon}-\pi_t-\varepsilon H_t\pi_t|
\,d\lambda
&=o(\varepsilon),\notag\\
\E_0\int
|\pi_{t,\varepsilon}-\widetilde\pi_{t,\varepsilon}|
\,d\lambda
&=o(\varepsilon).
\label{eq:overlap_proxy}
\end{align}

\paragraph{Step 2. Compare the policy values.}
Let $\widetilde J(\varepsilon)$ denote the auxiliary policy value.
Since $\widetilde\pi_{t,\varepsilon}\ll\pi_t$ at every stage,
its law is identified by the $g$-formula. Define
\[
w_{t,\varepsilon}
=
\frac{\widetilde\pi_{t,\varepsilon}(A_t\mid S_t)}
{\pi_t(A_t\mid S_t)},
\qquad
\Lambda_{t,\varepsilon}
=
\prod_{j<t}w_{j,\varepsilon}.
\]
Here $\Lambda_{t,\varepsilon}$ is the auxiliary history likelihood ratio, and $\E_0[w_{t,\varepsilon}\mid S_t]=1$.

We first show that 
\begin{equation}
\E_0\left|
\Lambda_{t,\varepsilon}-1-\varepsilon\sum_{j<t}H_j
\right|
=o(\varepsilon).
\label{eq:overlap_likelihood}
\end{equation}
The result holds at $t=1$. For the induction step,
\begin{align*}
\Lambda_{t+1,\varepsilon}-1-\varepsilon\sum_{j\leq t}H_j
={}&
\left(\Lambda_{t,\varepsilon}-1-\varepsilon\sum_{j<t}H_j\right)
w_{t,\varepsilon}\\
&+
(w_{t,\varepsilon}-1-\varepsilon H_t)\\
&+
\varepsilon\left(\sum_{j<t}H_j\right)(w_{t,\varepsilon}-1).
\end{align*}
Since $\Lambda_{t,\varepsilon}-1-\varepsilon\sum_{j<t}H_j$
is $S_t$-measurable and
$\E_0[w_{t,\varepsilon}\mid S_t]=1$, the induction hypothesis gives
\[
\E_0\left[
\left|\Lambda_{t,\varepsilon}-1-\varepsilon\sum_{j<t}H_j\right|
w_{t,\varepsilon}
\right]
=
\E_0\left|
\Lambda_{t,\varepsilon}-1-\varepsilon\sum_{j<t}H_j
\right|
=o(\varepsilon).
\]
Further, by equation~\eqref{eq:overlap_proxy} 
\[
\E_0|w_{t,\varepsilon}-1-\varepsilon H_t|
=
\E_0\int
|\widetilde\pi_{t,\varepsilon}-\pi_t-\varepsilon H_t\pi_t|
\,d\lambda
=o(\varepsilon).
\]
For the remaining term,
\[
\E_0[|w_{t,\varepsilon}-1|\mid S_t]\leq2,
\qquad
\E_0|w_{t,\varepsilon}-1|
\leq
\varepsilon\E_0|H_t|+o(\varepsilon)
=O(\varepsilon).
\]
Hence for every $K>0$,
\begin{align*}
\E_0\left[
\left|\sum_{j<t}H_j\right||w_{t,\varepsilon}-1|
\right]
\leq{}&
K\E_0|w_{t,\varepsilon}-1|+
2\E_0\left[
\left|\sum_{j<t}H_j\right|
\mathbf1\left\{\left|\sum_{j<t}H_j\right|>K\right\}
\right].
\end{align*}
Because $T$ is fixed and each $H_j$ is integrable, taking
$\varepsilon\downarrow0$ and then $K\uparrow\infty$ proves
\eqref{eq:overlap_likelihood}.

Next, set
\[
d_{t,\varepsilon}(s)
=
\|\pi_{t,\varepsilon}(\cdot\mid s)
-\widetilde\pi_{t,\varepsilon}(\cdot\mid s)\|_{\mathrm{TV}}.
\]
By \eqref{eq:overlap_proxy},
$0\leq d_{t,\varepsilon}\leq1$ and
$\E_0d_{t,\varepsilon}(S_t)=o(\varepsilon)$.
By \eqref{eq:overlap_likelihood} and
$0\leq d_{t,\varepsilon}\leq1$,
\begin{align*}
\E_{\widetilde\pi_\varepsilon}d_{t,\varepsilon}(S_t)
&=
\E_0[\Lambda_{t,\varepsilon}d_{t,\varepsilon}(S_t)]\\
&\leq
\E_0d_{t,\varepsilon}(S_t)
+
\varepsilon\E_0\left[
\left|\sum_{j<t}H_j\right|d_{t,\varepsilon}(S_t)
\right]\\
&\quad+
\underbrace{
\E_0\left|
\Lambda_{t,\varepsilon}-1-\varepsilon\sum_{j<t}H_j
\right|
}_{o(\varepsilon)}.
\end{align*}
For every $K>0$, the middle term satisfies
\begin{align*}
\E_0\left[
\left|\sum_{j<t}H_j\right|d_{t,\varepsilon}(S_t)
\right]
\leq{}&
K\E_0d_{t,\varepsilon}(S_t)+
\E_0\left[
\left|\sum_{j<t}H_j\right|
\mathbf1\left\{\left|\sum_{j<t}H_j\right|>K\right\}
\right].
\end{align*}
Taking the limits in the same order gives
\[
\E_{\widetilde\pi_\varepsilon}d_{t,\varepsilon}(S_t)
=o(\varepsilon).
\]

We next use a coupling argument to show that
$|J(\varepsilon)-\widetilde J(\varepsilon)|=o(\varepsilon)$.
Starting from the same initial history, couple the two trajectories using the same response realizations whenever their action histories
agree. 
At each common history $s$, maximally couple the action laws, giving disagreement probability $d_{t,\varepsilon}(s)$.
After the first disagreement, each trajectory follows its own policy.

Let $\Pr_{\mathrm c}$ and $\E_{\mathrm c}$ denote probability and
expectation under this coupling, and define the first disagreement time
\[
\tau
=
\inf\{t\leq T:
A_{t,\varepsilon}\ne\widetilde A_{t,\varepsilon}\},
\qquad \inf\varnothing=\infty.
\]
Writing $S_t$ for the auxiliary history, the two histories agree
on $\{\tau\geq t\}$. Hence
\begin{align*}
\Pr_{\mathrm c}(\tau\leq T)
&=
\sum_{t=1}^T
\E_{\mathrm c}\!\left[
\mathbf1\{\tau\geq t\}\,d_{t,\varepsilon}(S_t)
\right]\\
&\leq
\sum_{t=1}^T
\E_{\widetilde\pi_\varepsilon}
d_{t,\varepsilon}(S_t)
=o(\varepsilon).
\end{align*}

The returns agree on $\{\tau>T\}$.
The perturbed return is bounded by $C$ by assumption.
The auxiliary return inherits the same baseline bound because
its trajectory law is absolutely continuous with respect to
the baseline law. Therefore,
\begin{equation}
|J(\varepsilon)-\widetilde J(\varepsilon)|
\leq
2C\,\Pr_{\mathrm c}(\tau\leq T)
=o(\varepsilon).
\label{eq:overlap_valuecomparison}
\end{equation}

\paragraph{Step 3. Differentiate the auxiliary value.}
The likelihood representation gives
\[
\widetilde J(\varepsilon)-\widetilde J(0)
=
\E_0[\Gamma_1(\Lambda_{T+1,\varepsilon}-1)].
\]
By boundedness of $\Gamma_1$ and
\eqref{eq:overlap_likelihood},
\begin{align*}
&\left|
\widetilde J(\varepsilon)-\widetilde J(0)
-\varepsilon\E_0\left[\Gamma_1\sum_{t=1}^TH_t\right]
\right|\\
&\quad\leq
C\E_0\left|
\Lambda_{T+1,\varepsilon}-1-\varepsilon\sum_{t=1}^TH_t
\right|
=o(\varepsilon).
\end{align*}
Hence
\[
\widetilde J'(0+)
=
\E_0\left[\Gamma_1\sum_{t=1}^TH_t\right].
\]
Since $\sum_{k<t}\gamma^{k-1}Y_k$ is $S_t$-measurable and
$\E_0[H_t\mid S_t]=0$,
\begin{align*}
\E_0[\Gamma_1H_t]
&=
\E_0\left[
\left(\sum_{k<t}\gamma^{k-1}Y_k\right)
\E_0[H_t\mid S_t]
\right]
+
\gamma^{t-1}\E_0[H_t\Gamma_t]\\
&=
\gamma^{t-1}\E_0[H_t\Gamma_t].
\end{align*}
By the definition of $q_t$ and the identity $H_t\pi_t=L_t^*\pi_t$,
\[
\E_0[H_t\Gamma_t]
=
\E_0[H_tq_t(S_t,A_t)]
=
\E_0[\langle L_t^*\pi_t,q_t\rangle(S_t)].
\]
Summing over $t$ gives both forward expressions for $\widetilde J'(0+)$.

Further since $\widetilde J(0)=J(0)$,
\eqref{eq:overlap_valuecomparison} implies
\[
\frac{J(\varepsilon)-J(0)}{\varepsilon}
=
\frac{\widetilde J(\varepsilon)-\widetilde J(0)}{\varepsilon}
+
\frac{J(\varepsilon)-\widetilde J(\varepsilon)}{\varepsilon}
\longrightarrow
\widetilde J'(0+).
\]
Thus $\Theta=\widetilde J'(0+)$.

\paragraph{Step 4. Backward representation.}
Since $q_t\in\mathcal Q_t$, Lemma~1 gives
\[
\langle L_t^*\pi_t,q_t\rangle(S_t)
=
\E_0[(L_tq_t)(S_t,A_t)\mid S_t]
=
\langle L_tq_t,\pi_t\rangle(S_t)
\qquad\text{almost surely}.
\]
Taking expectations and substituting into the forward representation
yields both backward expressions.

Finally, the $L^1$ limit is unchanged by bounded modifications of $q_t$ outside the baseline distribution.
To see this, let $v$ be bounded with $v(S_t,A_t)=0$
almost surely. Then
\[
\int|v(S_t,a)|\pi_t(a\mid S_t)\,d\lambda(a)
=
\int|v(S_t,a)|H_t(S_t,a)\pi_t(a\mid S_t)\,d\lambda(a)
=0
\quad\text{almost surely}.
\]
Consequently,
\begin{align*}
&\E_0\left|
\frac{(M_{t,\varepsilon}v)(S_t,A_t)-v(S_t,A_t)}
{\varepsilon}
\right|\\
&\quad\leq
\frac1\varepsilon
\E_0\int|v(S_t,a)|\pi_{t,\varepsilon}(a\mid S_t)\,d\lambda(a)\\
&\quad=
\E_0\int|v(S_t,a)|
\left\{
\frac{\pi_{t,\varepsilon}-\pi_t}{\varepsilon}
-H_t\pi_t
\right\}(a\mid S_t)\,d\lambda(a)\\
&\quad\leq
\|v\|_\infty
\E_0\int\left|
\frac{\pi_{t,\varepsilon}-\pi_t}{\varepsilon}
-H_t\pi_t
\right|\,d\lambda
\longrightarrow0.
\end{align*}
\end{proof}
\end{revisionblock}

\subsection{Proof of Lemma~3}
\label{app:mixed_bias}

\begin{proof}
Fix $t$, and let $(q,H)$ satisfy the candidate conditions of the lemma.
Lemma~1 and first-order support give
\begin{equation}
G_t(q)=\gamma^{t-1}\E[H_t^\star(S_t,A_t)q(S_t,A_t)].
\label{eq:stagewise_score_form}
\end{equation}
Since $q_t^\star\in\mathcal Q_t$, the same identity holds for $q_t^\star$.
And $q_t^\star=\E[\Gamma_t\mid S_t,A_t]$,
\[
\E[H(S_t,A_t)\{\Gamma_t-q_t^\star(S_t,A_t)\}]=0.
\]
Therefore
\begin{align*}
\widetilde G_t(q,H)-G_t(q_t^\star)
&=\gamma^{t-1}\{\E[H_t^\star(q-q_t^\star)]+\E[H(\Gamma_t-q)]\}\\
&=\gamma^{t-1}\{\E[(H_t^\star-H)(q-q_t^\star)]
+\E[H(\Gamma_t-q_t^\star)]\}\\
&=\gamma^{t-1}\E[(H_t^\star-H)(q-q_t^\star)].
\end{align*}
This proves the mixed-bias identity and the two special cases.
For every admissible $q$, Cauchy--Schwarz gives
\[
|G_t(q)|\leq\gamma^{t-1}\|H_t^\star\|_2\|q\|_2.
\]
Thus $\gamma^{t-1}H_t^\star$ represents $G_t$ on $\mathcal Q_t$.
The weighted formula also defines a bounded linear extension to all of $L^2$, whose Riesz representer is unique: if $r\in L^2$ represents that extension,
choosing $q=r-\gamma^{t-1}H_t^\star$ yields $\E[(r-\gamma^{t-1}H_t^\star)^2]=0$.
\end{proof}

\subsection{Proof of Lemma~4}
\label{app:oracle_clt}

\begin{proof}
Define the oracle per-trajectory score by
\[
\psi^\star(Z)
:=
\sum_{t=1}^T
\left[
\gamma^{t-1}(L_t q_t^\star)(S_t,A_t)
+
\gamma^{t-1}H_t^\star(S_t,A_t)\,
\bigl\{\Gamma_t-q_t^\star(S_t,A_t)\bigr\}
\right].
\]
Then, by (23),
\[
\widehat\Theta^{\star}
=
\frac{1}{n}\sum_{i=1}^n \psi^\star(Z_i).
\]

We first show unbiasedness. By the definition of the stagewise functional in
(14),
\[
\E\!\left[
\gamma^{t-1}(L_t q_t^\star)(S_t,A_t)
\right]
=
G_t(q_t^\star).
\]
Next, since $H_t^\star(S_t,A_t)$ is measurable with respect to
$\sigma(S_t,A_t)$, iterated expectations give
\begin{align*}
\E\!\left[
\gamma^{t-1}H_t^\star(S_t,A_t)\,
\bigl\{\Gamma_t-q_t^\star(S_t,A_t)\bigr\}
\right]
&=
\gamma^{t-1}\E\!\left[
H_t^\star(S_t,A_t)\,
\E\!\left[
\Gamma_t-q_t^\star(S_t,A_t)
\,\middle|\,
S_t,A_t
\right]
\right] \\
&= 0.
\end{align*}
Therefore,
\[
\E\!\left[\psi^\star(Z)\right]
=
\sum_{t=1}^T G_t(q_t^\star).
\]
Hence we have 
\[
\E\!\left[\widehat\Theta^{\star}\right]=\Theta.
\]

Now let
\[
\sigma_\star^2
:=
\Var\!\big(\psi^\star(Z)\big).
\]
By assumption, the per-trajectory summand in
(23) has finite second moment, so
\[
\E\!\left[(\psi^\star(Z))^2\right]<\infty,
\qquad
\sigma_\star^2<\infty.
\]
Since the trajectories $\{Z_i\}_{i=1}^n$ are i.i.d., the random variables
$\{\psi^\star(Z_i)\}_{i=1}^n$ are also i.i.d. with mean $\Theta$ and variance
$\sigma_\star^2$. Therefore, by the central limit theorem,
\[
\sqrt n\,\bigl(\widehat\Theta^{\star}-\Theta\bigr)
=
\frac{1}{\sqrt n}\sum_{i=1}^n
\Bigl(
\psi^\star(Z_i)-\E[\psi^\star(Z)]
\Bigr)
\ \xrightarrow{d}\
\mathcal N(0,\sigma_\star^2).
\]
\end{proof}

\subsection{Proof of Theorem~5}
\label{app:dr_clt}

\begin{proof}
Condition on the training observations for one scoring fold of size $m$.
The fitted nuisances are then fixed, and the scoring trajectories are
independent with the baseline law. Write $\E_n$ for the fold average.
All expectations, variances, and $L^2$ norms below are conditional on
the training observations. Let
\[
\Delta q_t=\widehat q_t-q_t^\star,\qquad
\Delta H_t=\widehat H_t-H_t^\star,\qquad
e_t^\star=\Gamma_t-q_t^\star(S_t,A_t),
\]
and suppress arguments evaluated at $(S_t,A_t)$.
Then $\E[e_t^\star\mid S_t,A_t]=0$ and boundedness gives
$|e_t^\star|\vee|\Delta q_t|\le2\Gamma_{\max}$.

Expanding the augmentation term and using linearity of $L_t$ gives
$\widehat\Theta-\widehat\Theta^\star=A_n+B_n+C_n$, where
\begin{align}
A_n&:=\E_n\!\left[\sum_{t=1}^T
\gamma^{t-1}\Delta H_t e_t^\star\right],
\label{eq:An_def_revised_op}\\
B_n&:=\E_n\!\left[\sum_{t=1}^T
\gamma^{t-1}\{L_t\Delta q_t-H_t^\star\Delta q_t\}\right],
\label{eq:Bn_def_revised_op}\\
C_n&:=-\E_n\!\left[\sum_{t=1}^T
\gamma^{t-1}\Delta H_t\Delta q_t\right].
\label{eq:Cn_def_revised_op}
\end{align}

\medskip
\noindent\emph{Step 1: Control of $A_n$.}
Conditional centering of $e_t^\star$ gives $\E[A_n]=0$. Since $T$ is fixed,
\[
m\Var(A_n)
\le4T\Gamma_{\max}^2\sum_{t=1}^T
\gamma^{2(t-1)}\|\Delta H_t\|_2^2=o_p(1).
\]
Conditional Chebyshev's inequality implies $\sqrt m A_n=o_p(1)$.

\medskip
\noindent\emph{Step 2: Control of $B_n$.}
Since $\mathcal Q_t$ is linear, $\Delta q_t\in\mathcal Q_t$.
Applying Lemma~1 conditionally on the training sample
and using first-order support gives
$\E[L_t\Delta q_t]=\E[H_t^\star\Delta q_t]$, so $\E[B_n]=0$ and
\[
m\Var(B_n)
\le2T\sum_{t=1}^T\gamma^{2(t-1)}
\left\{\|L_t\Delta q_t\|_2^2
+\|H_t^\star\Delta q_t\|_2^2\right\}.
\]
The first term tends to zero by operator consistency. For any $M>0$,
\[
\|H_t^\star\Delta q_t\|_2^2
\le M^2\|\Delta q_t\|_2^2
+4\Gamma_{\max}^2\E\!\left[
(H_t^\star)^2\mathbf1\{|H_t^\star|>M\}\right].
\]
For fixed $M$, the first term is $o_p(1)$, while the second tends to
zero as $M\to\infty$ because $H_t^\star\in L^2$.
Thus $\|H_t^\star\Delta q_t\|_2=o_p(1)$ and
$\sqrt m B_n=o_p(1)$. No polynomial rate for the operator error is needed.

\medskip
\noindent\emph{Step 3: Control of $C_n$.}
Cauchy--Schwarz and the product-rate condition yield
\[
|\E[C_n]|
\le\sum_{t=1}^T\gamma^{t-1}
\|\Delta H_t\|_2\|\Delta q_t\|_2=o_p(n^{-1/2}).
\]
Moreover,
\[
m\Var(C_n)
\le4T\Gamma_{\max}^2\sum_{t=1}^T
\gamma^{2(t-1)}\|\Delta H_t\|_2^2=o_p(1).
\]
Consequently, $\sqrt m\{C_n-\E[C_n]\}=o_p(1)$.
For a fold of size $m$, multiplication by its full-sample weight
$m/n$ gives
$\sqrt n(m/n)o_p(m^{-1/2})=\sqrt{m/n}\,o_p(1)$.
The number of folds is fixed, and the population terms are
$o_p(n^{-1/2})$ on each fold. Hence
\[
\sqrt n(\widehat\Theta-\widehat\Theta^\star)=o_p(1).
\]
Lemma~4 and Slutsky's theorem establish the asserted normal limit.

\medskip
\noindent\emph{Consistency of the variance estimator.}
Let $\widehat\psi_i$ be the out-of-fold trajectory contribution and
$\psi^\star$ the oracle contribution from Lemma~4.
For a new scoring trajectory, the preceding bounds imply, on each fold,
\[
\|\widehat\psi-\psi^\star\|_2
\le\sum_{t=1}^T\gamma^{t-1}
\left\{\|L_t\Delta q_t\|_2
+2\Gamma_{\max}\|\Delta H_t\|_2
+\|H_t^\star\Delta q_t\|_2\right\}=o_p(1).
\]
Conditional Markov's inequality, summed over the fixed number of folds, gives
\[
\frac1n\sum_{i=1}^n
\{\widehat\psi_i-\psi^\star(Z_i)\}^2=o_p(1).
\]
Cauchy--Schwarz and the oracle law of large numbers therefore imply
$n^{-1}\sum_i\widehat\psi_i^2\to_p\E[\psi^\star(Z)^2]$.
Since $\widehat\Theta\to_p\Theta$,
\[
\widehat\sigma^2
:=\frac1{n-1}\sum_{i=1}^n
(\widehat\psi_i-\widehat\Theta)^2
=\frac n{n-1}\left\{
\frac1n\sum_{i=1}^n\widehat\psi_i^2-\widehat\Theta^2\right\}
\to_p\sigma_\star^2.
\]
Thus $\widehat{\operatorname{SE}}^2=\widehat\sigma^2/n$ is consistent
on the oracle variance scale. If $\sigma_\star^2>0$, the studentized
estimator converges to $N(0,1)$.
\end{proof}

\subsection{Proof of Proposition~6}
\label{app:horizon_variance}

\begin{revisionblock}
\begin{proof}
Write $\kappa_T=\sum_{j=0}^{T-1}\gamma^j$ and
$D_t=L_tq_t^\star-H_t^\star q_t^\star$, with stage functions
evaluated at $(S_t,A_t)$. The oracle contribution is
\[
\psi_T^\star
=
\sum_{t=1}^T\gamma^{t-1}D_t
+
\sum_{t=1}^T\gamma^{t-1}H_t^\star\Gamma_t.
\]
The identity $\int L_t^*\pi_t\,d\lambda=0$ and
Lemma~1 imply
\begin{align*}
\E[H_t^\star\mid S_t]
&=
\int(L_t^*\pi_t)(a\mid S_t)\,d\lambda(a)=0,\\
\E[D_t\mid S_t]
&=
\langle L_tq_t^\star,\pi_t\rangle(S_t)
-
\langle L_t^*\pi_t,q_t^\star\rangle(S_t)=0.
\end{align*}
For $t<u$, both $H_t^\star$ and $D_t$ are $S_u$-measurable.
Therefore,
\[
\E[H_t^\star H_u^\star]
=
\E[H_t^\star\E[H_u^\star\mid S_u]]=0,
\qquad
\E[D_tD_u]
=
\E[D_t\E[D_u\mid S_u]]=0.
\]

Since $|q_t^\star|\leq Y_{\max}\kappa_T$, 
\[
\|D_t\|_2
\leq
\|L_tq_t^\star\|_2+\|H_t^\star q_t^\star\|_2
\leq
(C_L+C_H)Y_{\max}\kappa_T.
\]
Using the vanishing cross moments and
$\sum_{t=1}^T\gamma^{2(t-1)}\leq \kappa_T$, we obtain
\begin{align*}
\left\|\sum_{t=1}^T\gamma^{t-1}D_t\right\|_2^2
&=
\sum_{t=1}^T\gamma^{2(t-1)}\|D_t\|_2^2\\
&\leq
(C_L+C_H)^2Y_{\max}^2\kappa_T^3.
\end{align*}

For the second component, rearranging the finite sums gives
\[
\sum_{t=1}^T\gamma^{t-1}H_t^\star\Gamma_t
=
\sum_{u=1}^T\gamma^{u-1}Y_u
\sum_{t=1}^uH_t^\star.
\]
The preceding cross-moment identity implies
\[
\left\|\sum_{t=1}^uH_t^\star\right\|_2^2
=
\sum_{t=1}^u\|H_t^\star\|_2^2
\leq uC_H^2.
\]
Thus, by the outcome bound and the triangle inequality,
\[
\left\|
\sum_{t=1}^T\gamma^{t-1}H_t^\star\Gamma_t
\right\|_2
\leq
Y_{\max}C_H\sum_{u=1}^T\gamma^{u-1}\sqrt u.
\]
Cauchy--Schwarz gives
\[
\left(\sum_{u=1}^T\gamma^{u-1}\sqrt u\right)^2
\leq
\kappa_T\sum_{u=1}^T u\gamma^{u-1}
\leq \kappa_T^3,
\]
where
\[
\sum_{u=1}^T u\gamma^{u-1}
=
\sum_{\substack{r,v\geq0\\r+v<T}}\gamma^{r+v}
\leq
\left(\sum_{j=0}^{T-1}\gamma^j\right)^2
=\kappa_T^2.
\]
Combining the two components yields
\[
\|\psi_T^\star\|_2
\leq
Y_{\max}(C_L+2C_H)\kappa_T^{3/2},
\]
and hence
\[
\Var(\psi_T^\star)
\leq
\|\psi_T^\star\|_2^2
\leq
Y_{\max}^2(C_L+2C_H)^2\kappa_T^3.
\]
Since the oracle estimator averages $n$ independent trajectory
contributions,
\[
\Var(\widehat\Theta^\star)
=
\frac{\Var(\psi_T^\star)}{n}.
\]
The stated horizon bounds follow from $\kappa_T=T$ when $\gamma=1$
and $\kappa_T\leq(1-\gamma)^{-1}$ when $0<\gamma<1$.
\end{proof}
\end{revisionblock}

\section{Experimental details}
\label{app:experimental-details}

\subsection{Dynamic pricing simulator}
\label{app:simulator-details}

This subsection gives the details of the simulator used in
Section~4.2. Demand has three components: seasonality,
a direct effect of current price on purchase probability, and a reference-price
channel through which current prices affect future demand.

\paragraph{Observed history and latent state.}
At each period $t=1,\ldots,T$, the platform posts a price
$A_{i,t}\in[p_{\min},p_{\max}]$ for consumer $i$. The observed state is
\[
X_{i,t}
=
\Bigl(
P^{\mathrm{last}}_{i,t},\
\bar P_{i,t},\
Y^{\mathrm{last}}_{i,t},\
\bar Y_{i,t},\
\sin(2\pi t/T),\
\cos(2\pi t/T)
\Bigr),
\]
and the full observed history is
\[
S_{i,t}
=
(X_{i,1},A_{i,1},Y_{i,1},\ldots,
X_{i,t-1},A_{i,t-1},Y_{i,t-1},X_{i,t}).
\]

Each consumer has time-invariant latent heterogeneity $U_i=(W_i,\alpha_i)$ where $W_i$ governs baseline willingness to pay and $\alpha_i$ controls how quickly the consumer updates an internal reference level. We draw
\[
W_i\sim N(0,\sigma_W^2),
\]
and
\[
\alpha_i
=
\alpha_{\min}
+
(\alpha_{\max}-\alpha_{\min})
\frac{1}{1+\exp(-c_\alpha \xi_i^\alpha)},
\qquad
\xi_i^\alpha\sim N(0,1).
\]
Thus $\alpha_i\in(\alpha_{\min},\alpha_{\max})$, with larger values
corresponding to slower updating of the internal reference price.

In addition to $(W_i,\alpha_i)$, the consumer has two latent time-varying states:
\[
r^\star_{i,t}
\quad\text{(internal reference price)},
\qquad
u_{i,t}
\quad\text{(persistent taste shock)}.
\]
Initial conditions are
\[
r^\star_{i,1}
=
\operatorname{clip}\!\left(
r_0+\sigma_{r,0}\xi_{i,0}+\pi_W W_i+\sigma_{\mathrm{init}}\zeta_{i,0},
[p_{\min},p_{\max}]
\right),
\qquad
\xi_{i,0},\zeta_{i,0}\stackrel{\mathrm{iid}}{\sim}N(0,1),
\]
and
\[
u_{i,1}
=
0.35W_i+0.25\varepsilon^u_{i,0},
\qquad
\varepsilon^u_{i,0}\sim N(0,1).
\]
The observed summaries are initialized as
\[
P^{\mathrm{last}}_{i,1}=\bar P_{i,1}=r_0,
\qquad
Y^{\mathrm{last}}_{i,1}=\bar Y_{i,1}=0.
\]

\paragraph{Baseline policy and local perturbation.}
Write
\[
\operatorname{clip}(x,[\ell,u])
:=
\min\{u,\max\{\ell,x\}\}.
\]
Conditional on the observed history, the baseline pricing policy is
\[
A^{\mathrm{raw}}_{i,t}\mid S_{i,t}
\sim
N\bigl(\mu_t(X_{i,t}),\sigma_{\mathrm{price}}^2\bigr),
\qquad
A_{i,t}
=
\operatorname{clip}\bigl(A^{\mathrm{raw}}_{i,t},[p_{\min},p_{\max}]\bigr).
\]
The policy mean depends on bounded summaries of recent prices and realized outcomes. Define
\[
\widetilde Y^{\mathrm{last}}_{i,t}
=
\tanh\!\left(\frac{Y^{\mathrm{last}}_{i,t}-\mu_Y}{s_Y}\right),
\qquad
\widetilde{\bar Y}_{i,t}
=
\tanh\!\left(\frac{\bar Y_{i,t}-\mu_Y}{s_Y}\right),
\]
and
\[
\widetilde G_{i,t}
=
\tanh\!\left(\frac{P^{\mathrm{last}}_{i,t}-\bar P_{i,t}}{2}\right).
\]
Then
\[
\begin{aligned}
\mu_t(X_{i,t})=\operatorname{clip}\!\Bigl(&c_0
+c_1P^{\mathrm{last}}_{i,t}+c_2\bar P_{i,t}
+c_3\widetilde Y^{\mathrm{last}}_{i,t}
+c_4\widetilde{\bar Y}_{i,t}\\
&+c_5\widetilde G_{i,t}+c_6\sin(2\pi t/T)+c_7\cos(2\pi t/T),
[p_{\min}+m,p_{\max}-m]\Bigr),
\end{aligned}
\]
capturing pricing inertia, feedback from recent outcomes, and seasonality.

The local perturbation is a uniform location shift of the latent Gaussian pricing rule before clipping:
\[
A^{\mathrm{raw}}_{i,t,\varepsilon}\mid S_{i,t}
\sim
N\bigl(\mu_t(X_{i,t})+\varepsilon,\sigma_{\mathrm{price}}^2\bigr),
\qquad
A_{i,t,\varepsilon}
=
\operatorname{clip}\bigl(A^{\mathrm{raw}}_{i,t,\varepsilon},
[p_{\min},p_{\max}]\bigr).
\]

\begin{revisionblock}
Clipping keeps prices within the feasible range $[p_{\min},p_{\max}]$.
Before and after a finite Gaussian mean shift, the policy has a positive density on $(p_{\min},p_{\max})$ and positive probability at both endpoints.
The shifted policy is therefore absolutely continuous with respect to the baseline policy.
Writing $C(z)=\operatorname{clip}(z,[p_{\min},p_{\max}])$ and
$Z\mid S_t=s\sim N(\mu_t(X_t),\sigma_{\mathrm{price}}^2)$, the observed-action
kernel is
\[
M_{t,\varepsilon}(\cdot\mid a,s)
=
\mathcal L\{C(Z+\varepsilon)\mid C(Z)=a,S_t=s\}.
\]
An interior action has $Z=a$. A positive shift leaves the upper atom fixed.
At the lower atom, an $O(\varepsilon)$ fraction of conditional latent draws crosses the bound, each moving by at most $\varepsilon$.
For locally Lipschitz candidates, the conditional change is therefore
$O(\varepsilon^2)$. The backward derivative is zero at both atoms.
\end{revisionblock}

\paragraph{Demand, outcome, and hidden-state transitions.}
Conditional on $(S_{i,t},A_{i,t})$ and the latent state, the purchase indicator
$B_{i,t}\in\{0,1\}$ is drawn as
\[
B_{i,t}\mid(S_{i,t},A_{i,t},W_i,\alpha_i,r^\star_{i,t},u_{i,t})
\sim
\mathrm{Bernoulli}\!\bigl(\sigma(\ell_{i,t})\bigr),
\qquad
\sigma(x)=\frac{1}{1+e^{-x}}.
\]

To make the reference channel explicit, define
\[
\Delta_{i,t}:=r^\star_{i,t}-A_{i,t}.
\]
The demand index is
\[
\begin{aligned}
\ell_{i,t}={}&\beta_0+W_i+u_{i,t}
+\beta_{\mathrm{season}}
\{0.9\sin(2\pi t/T)-0.5\cos(2\pi t/T)\}\\
&+\beta_{p,i}A_{i,t}+\beta_{r,i}\Delta_{i,t}
-\beta_{\mathrm{gap},2}\Delta_{i,t}^2
-\beta_{\mathrm{nl}}\sin(\Delta_{i,t}),
\end{aligned}
\]
where
\[
\beta_{p,i}
=
\beta_p\exp(-0.22W_i),
\]
and
\[
\beta_{r,i}
=
\beta_r
\left\{
1+
1.6\,
\frac{\alpha_i-\bar\alpha}{\alpha_{\max}-\alpha_{\min}}
\right\},
\qquad
\bar\alpha
=
\frac{\alpha_{\min}+\alpha_{\max}}{2}.
\]

Here $W_i$ shifts willingness to pay and price sensitivity through
$\beta_{p,i}$, while $\alpha_i$ enters the reference channel through
$\beta_{r,i}$. A larger gap $\Delta_{i,t}=r^\star_{i,t}-A_{i,t}$ raises demand through the linear reference term, while the quadratic and sinusoidal terms allow for nonlinear and asymmetric reference effects.

The period outcome is realized revenue,
\[
Y_{i,t}=A_{i,t}B_{i,t},
\]
and the discounted cumulative outcome is
\[
\Gamma_{i,t}
=
\sum_{u=t}^T\gamma^{u-t}Y_{i,u}.
\]

The hidden states evolve according to
\begin{align}
r^\star_{i,t+1}
&=
\operatorname{clip}\!\left(
\alpha_i r^\star_{i,t}
+
(1-\alpha_i)A_{i,t}
+
\sigma_r\eta_{i,t},
[p_{\min},p_{\max}]
\right),
\qquad
\eta_{i,t}\stackrel{\mathrm{iid}}{\sim}N(0,1),
\label{eq:rstar_transition_appendix_hidden}
\\
u_{i,t+1}
&=
\pi_u u_{i,t}
+
\lambda_u B_{i,t}
+
\sigma_u\varepsilon^u_{i,t},
\qquad
\varepsilon^u_{i,t}\stackrel{\mathrm{iid}}{\sim}N(0,1).
\label{eq:latent_taste_transition_appendix_hidden}
\end{align}
Thus the current price affects future demand through the update of the latent reference price $r^\star_{i,t+1}$, while $u_{i,t}$ captures persistent, purchase-dependent taste variation.

The observed summaries are updated deterministically from realized prices and outcomes:
\[
P^{\mathrm{last}}_{i,t+1}=A_{i,t},
\qquad
Y^{\mathrm{last}}_{i,t+1}=Y_{i,t},
\]
\[
\bar P_{i,t+1}
=
\frac{(t-1)\bar P_{i,t}+A_{i,t}}{t},
\qquad
\bar Y_{i,t+1}
=
\frac{(t-1)\bar Y_{i,t}+Y_{i,t}}{t}.
\]

\paragraph{Initialization and default parameters.}
Unless otherwise stated, we use
\[
T=8,
\qquad
n=5000,
\qquad
\gamma=0.99,
\qquad
(p_{\min},p_{\max})=(1,10),
\qquad
\sigma_{\mathrm{price}}=0.70.
\]
For latent heterogeneity and the hidden-reference process, the default values are
\[
\sigma_W=1.00,
\qquad
(\alpha_{\min},\alpha_{\max},c_\alpha)=(0.72,0.97,1.10),
\]
and
\[
(r_0,\sigma_{r,0},\pi_W,\sigma_{\mathrm{init}},\sigma_r)
=
(5.20,0.85,0.60,0.35,0.12).
\]
For the policy mean, we use
\[
(c_0,c_1,c_2,c_3,c_4,c_5,c_6,c_7)
=
(2.90,0.25,0.18,0.28,0.18,0.15,0.18,-0.10),
\]
with
\[
(\mu_Y,s_Y,m)=(2.70,2.20,0.55).
\]
For demand and latent taste dynamics, we use
\[
(\beta_0,\beta_p,\beta_r,\beta_{\mathrm{gap},2},
\beta_{\mathrm{season}},\beta_{\mathrm{nl}})
=
(2.00,-0.72,0.78,0.18,0.20,0.14),
\]
and
\[
(\pi_u,\sigma_u,\lambda_u)=(0.88,0.22,0.40).
\]

\paragraph{Policy score and backward operator.}
Because clipping creates atoms at $p_{\min}$ and $p_{\max}$, the observed baseline policy is mixed discrete--continuous. Let
\[
z_{L,t}(X_{i,t})
=
\frac{p_{\min}-\mu_t(X_{i,t})}{\sigma_{\mathrm{price}}},
\qquad
z_{U,t}(X_{i,t})
=
\frac{p_{\max}-\mu_t(X_{i,t})}{\sigma_{\mathrm{price}}},
\]
and let $\phi$ and $\Phi$ denote the standard normal pdf and cdf.
The oracle score for the clipped Gaussian mean-shift path is
\[
H_t^\star(X_{i,t},A_{i,t})
=
\begin{cases}
\dfrac{A_{i,t}-\mu_t(X_{i,t})}{\sigma_{\mathrm{price}}^2},
& p_{\min}<A_{i,t}<p_{\max},\\[2.5mm]
-\dfrac{\phi(z_{L,t}(X_{i,t}))}
{\sigma_{\mathrm{price}}\Phi(z_{L,t}(X_{i,t}))},
& A_{i,t}=p_{\min},\\[3mm]
\dfrac{\phi(z_{U,t}(X_{i,t}))}
{\sigma_{\mathrm{price}}\{1-\Phi(z_{U,t}(X_{i,t}))\}},
& A_{i,t}=p_{\max}.
\end{cases}
\]
\revision{This score is used for the ASRW (oracle score) benchmark
and the reference checks.}
The feasible ASRW estimator instead learns $\widehat H_t$ using the
variational Riesz objective from Section~4.

\revision{Both the direct term and score-learning objective use the
backward operator}
\[
(L_t h)(s,a)
=
\mathbf 1\{p_{\min}<a<p_{\max}\}\,\partial_a h(s,a).
\]
Equivalently, for sufficiently smooth $h$,
\[
\E\!\left[
H_t^\star(X_{i,t},A_{i,t})\,h(S_{i,t},A_{i,t})
\right]
=
\E\!\left[
\mathbf 1\{p_{\min}<A_{i,t}<p_{\max}\}\,
\partial_a h(S_{i,t},A_{i,t})
\right].
\]

\begin{figure}[t]
    \centering
    \includegraphics[width=0.8\linewidth]{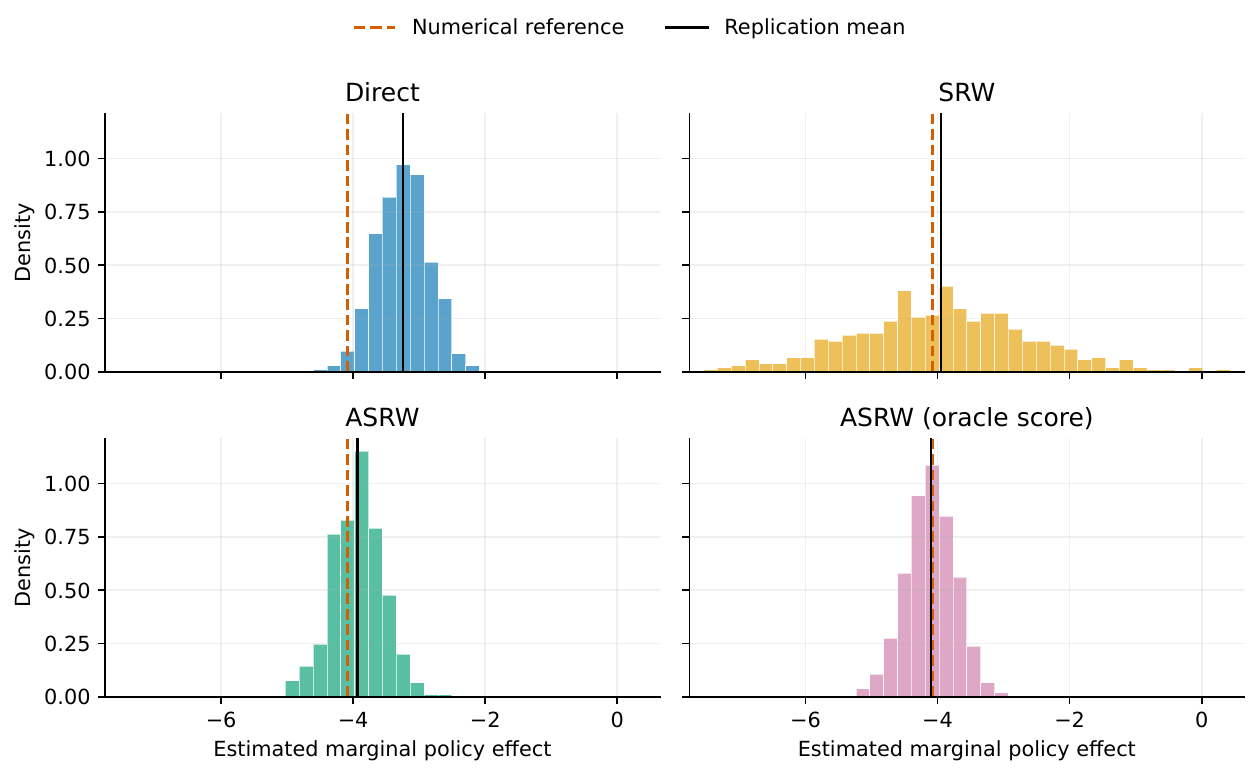}
    \caption{
    Sampling distributions of the four \revision{IPE} estimators in the
    dynamic pricing simulator over \revision{500 independent replications}.
    \revision{Dashed lines mark the reference value; solid lines mark
    the replication means.}
    }
    \label{fig:pricing_histograms}
\end{figure}

\subsection{Numerical implementation}
\label{app:numerical-implementation}

\begin{revisionblock}
\paragraph{Reference values.}
In both designs, we approximate the right derivative by
\[
\widehat\Theta_{\rm ref}(\delta)
=
\frac{1}{N_{\rm MC}}
\sum_{i=1}^{N_{\rm MC}}
\frac{G_i(\delta)-G_i(0)}{\delta},
\qquad
G_i(\varepsilon)
=
\sum_{t=1}^T\gamma^{t-1}Y_{i,t}^{(\varepsilon)},
\quad \delta>0.
\]
We use $\delta=0.0125$, with
$N_{\rm MC}=5{,}000{,}000$ for each benchmark reference distribution,
including the horizon study, and $15{,}000{,}000$ for pricing.
Monte Carlo standard errors are computed from trajectory contributions.

\paragraph{Nuisance estimation and the oracle-score comparison.}
Both nuisance learners use the full observed history and current action,
with five-fold cross-fitting. Both are feedforward networks with SiLU
activations, trained using mini-batch AdamW.
The benchmark uses two hidden layers, with $128$ units per layer for
outcome regression and $96$ for score estimation.
Pricing uses two hidden layers of $128$ units for outcome regression and
three hidden layers of $128$ units for score estimation, pooling stages
for both learners.
Early stopping uses validation prediction loss for outcome regression
and the validation Riesz objective for score estimation.
Outcome predictions are clipped to the known bounds on $\Gamma_t$.
Action derivatives are computed using \texttt{torch.autograd.grad}.
In pricing, endpoint observations remain in the squared term of the
Riesz objective.

The comparator ASRW (oracle score) is
\[
\begin{aligned}
\widehat\Theta^{\rm oracle\ score}
={}&\frac1n\sum_{i=1}^n\sum_{t=1}^T\gamma^{t-1}
\Bigl[
(L_t\widehat q_t^{(-k(i))})(S_{i,t},A_{i,t})\\
&\qquad+
H_t^\star(S_{i,t},A_{i,t})
\{\Gamma_{i,t}-\widehat q_t^{(-k(i))}(S_{i,t},A_{i,t})\}
\Bigr],
\end{aligned}
\]
where $k(i)$ is the fold containing trajectory $i$.
It retains ASRW's fitted outcome regressions and folds, replacing only
the learned scores to isolate the effect of score estimation.

\paragraph{Standard errors.}
For each estimator, let $\widehat\psi_i$ be its out-of-fold trajectory
contribution, so that $\widehat\Theta=n^{-1}\sum_i\widehat\psi_i$.
The standard errors are calculated via
\[
\widehat{\operatorname{SE}}
=
\left\{
\frac{1}{n(n-1)}\sum_{i=1}^n
(\widehat\psi_i-\widehat\Theta)^2
\right\}^{1/2},
\]
with $95\%$ intervals constructed as
$\widehat\Theta\pm1.96\widehat{\operatorname{SE}}$.
Stagewise contributions are summed within each trajectory before
computing the sample variance.
For ASRW, consistency of this variance estimate is established in the
proof of Theorem~5 in Section~\ref{app:dr_clt}.
For Direct and SRW, we report empirical coverage using the same
interval construction.

\paragraph{Replications and horizon study.}
The benchmark uses $1000$ independent replications per configuration,
and pricing uses $500$.
The horizon study uses the same DGP as the simulation, with
$T\in\{2,3,5,10,15,20,30\}$, $n=1000$, $p=5$, $\gamma\in\{1,0.99\}$, and $200$ independent replications per configuration.
\end{revisionblock}

\end{document}